\documentclass[lettersize,journal]{IEEEtran}
\usepackage{amsmath,amsfonts}
\usepackage{amssymb}
\usepackage{algorithm}
\usepackage{array}
\usepackage{subfigure}
\usepackage[caption=false,font=normalsize,labelfont=sf,textfont=sf]{subfig}
\usepackage{textcomp}
\usepackage{stfloats}
\usepackage{url}
\usepackage{verbatim}
\usepackage{graphicx}
\usepackage{cite}
\usepackage{algpseudocode}
\usepackage{diagbox}
\usepackage{booktabs}
\usepackage{subfigure}
\usepackage{xcolor}
\usepackage{amsthm}
\usepackage{xcolor}
\usepackage{xpatch}
\usepackage{caption}

\makeatletter
\def\changeBibColor#1{%
	\in@{#1}{}
	\ifin@\color{blue}\else\normalcolor\fi
}

\xpatchcmd\@bibitem
{\item}
{\changeBibColor{#1}\item}
{}{\fail}

\xpatchcmd\@lbibitem
{\item}
{\changeBibColor{#2}\item}
{}{\fail}
\makeatother

\begin{document}

\title{Age of Information Aided Intelligent  Grant-Free Massive Access for Heterogeneous mMTC Traffic}

\author{Zhongwen Sun, Wei Chen, \IEEEmembership{Senior Member, IEEE}, Yuxuan Sun, \IEEEmembership{Member, IEEE}, Bo Ai, \IEEEmembership{Fellow, IEEE}
\thanks{Zhongwen Sun, Wei Chen, and Bo Ai are with the State Key Laboratory of Advanced Rail Autonomous Operation and the School of Electronic and Information Engineering, Beijing Jiaotong University, Beijing 100044, China (e-mail: 22120121@bjtu.edu.cn; weich@bjtu.edu.cn; boai@bjtu.edu.cn). Yuxuan Sun is with the School of Electronic and Information Engineering, Beijing Jiaotong University, Beijing 100044, China (e-mail: yxsun@bjtu.edu.cn).
	
}
\thanks{Part of this work has been presented in IEEE MLSP 2024\cite{10734800}.}}



\maketitle

\begin{abstract}
With the arrival of 6G, the Internet of Things (IoT) traffic is becoming more and more complex and diverse. To meet the diverse service requirements of IoT devices, massive machine-type communications (mMTC) becomes a typical scenario, and more recently, grant-free random access (GF-RA) presents a promising direction due to its low signaling overhead. However, existing GF-RA research primarily focuses on improving the accuracy of user detection and data recovery, without considering the heterogeneity of traffic. In this paper, we investigate a non-orthogonal GF-RA scenario where two distinct types of traffic coexist: event-triggered traffic with alarm devices (ADs), and status update traffic with monitor devices (MDs). The goal is to simultaneously achieve high detection success rates for ADs and high information timeliness for MDs. First, we analyze the age-based random access scheme and optimize the access parameters to minimize the average age of information (AoI) of MDs. Then, we design an age-based prior information aided autoencoder (A-PIAAE) to jointly detect active devices, together with learned pilots used in GF-RA to reduce interference between non-orthogonal pilots. In the decoder, an Age-based Learned Iterative Shrinkage Thresholding Algorithm (LISTA-AGE) utilizing the AoI of MDs as the prior information is proposed to enhance active user detection. Theoretical analysis is provided to demonstrate the proposed A-PIAAE has better convergence performance. Experiments demonstrate the advantage of the proposed method in reducing the average AoI of MDs and improving the successful detection rate of ADs. 
\end{abstract}

\begin{IEEEkeywords}
Age of information, grant-free random access, mMTC, heterogeneous traffic, deep learning.
\end{IEEEkeywords}

\section{Introduction}
By 2030, digitalization and data-centric approaches are expected to transform our societies, driven by key verticals such as connected industries, intelligent transport systems and smart cities \cite{2004.14146}. Massive machine-type communication (mMTC) (or massive communication in 6G) is one of the main drivers of this digitization. Moreover, 6G will be an agile and efficient convergent network serving a set of diverse service classes. MMTC-specific 6G key performance indicators (KPI) will be much more stringent than those considered for 5G, and include a diverse set of novel metrics not considered before. For example, smart home applications require reliable real-time operation guarantees and energy efficiency across distributed nodes; vehicle monitoring applications demand quality assurances for real-time data monitoring \cite{qos-difference}.

Generally, there are two typical types of traffic, i.e., event-triggered traffic and status update traffic, which coexist in mMTC. Some devices function as alarm devices (ADs), responsible for transmitting urgent information such as fire and earthquake alerts. This will result in event-triggered traffic, which requires high probability of successful access. For status update traffic, the freshness of uploaded information from monitor devices (MDs) is crucial for updating the system status in applications like autonomous vehicles, smart agriculture and smart city. In those application, the freshness of information is critical for environment perception, real-time monitoring and control. For example, within the realm of smart agriculture \cite{ABBAS2023199}, the freshness of information including humidity, temperature and moisture are indispensable for the effective operation of monitoring systems. To measure the information freshness at the receiver, a new metric called age of information (AoI) is introduced \cite{5984917}, which is a function of both how often packets are transmitted and how much delay packets experience in the communication network.

The prerequisite for supporting these traffic is the successful and efficient access of the massive devices in mMTC. Due to the characteristics of small packet transmission, the signaling interactions in the traditional four-step random access \cite{3gpp_ts_38_211} will incur huge signaling overheads, seriously affecting the spectral efficiency. A prospective solution is grant-free random access (GF-RA), where channel resources are accessed by active devices without obtaining permission from the base station (BS). This absence of scheduling, however, leads the BS to receive a composite signal from all active devices, making it challenging to identify individual devices and recover their signals. By leveraging the sporadic communication characteristics in mMTC scenario, which means that only a small number of devices are active at the same time among the huge total number of devices, the channel estimation, active user detection (AUD) and data recovery problems can be formulated as an under-determined sparse linear inverse problem \cite{6125356,bockelmann2013compressive,6629742}. The problem can then be effectively solved using compressed sensing (CS) algorithms such as the iterative shrinkage thresholding algorithm (ISTA) \cite{doi:10.1137/080716542}, approximate message passing (AMP) \cite{5503193} and sparse Bayesian learning\cite{1315936}.

Different from the previous work in \cite{10734800} that only considered the status update traffic and minimized the average AoI,
in this paper, we consider the massive access problem for heterogeneous traffic in mMTC scenarios, in which ADs that require high successful detection rate and MDs that require high information freshness coexist. To support massive devices with varying access requirements, we first enable heterogeneous traffic devices to share pilot resources in a non-orthogonal manner. To minimize the average AoI of MDs, an optimized age-based random access (ARA) scheme is designed. Subsequently, AoI is utilized as prior information to enhance AUD performance. To further reduce interference between the two types of devices, an autoencoder is designed to optimize the non-orthogonal pilots of devices.
The main contributions of this paper are summarized as follows: 
\begin{itemize}
	
	\item[1)] We consider heterogeneous traffic in mMTC, with the goal of achieving a high successful detection rate for ADs and high information timeliness for MDs simultaneously. We first optimize the age threshold and access probability to reduce the average AoI of MDs. Leveraging CS theory, we derive the successful detection rate and derive the average AoI of MDs as a function of the threshold and access probability. Then a two-dimensional search is used to determine the optimal access parameters.
	\item[2)] We design an age-based prior information aided autoencoder (A-PIAAE) to jointly reduce the system average AoI and improve the successful detection rate, together with learned pilot used in GF-RA. The encoder is a linear encoder through the pilot matrix, which is trained to reduce interference between devices. The decoder is a deep-unfolding NN called LISTA-AGE, using the prior information that devices with an AoI less than the optimal threshold are definitely not active under the optimized ARA scheme as an additional input to assist the AUD.
	\item[3)] We conduct theoretical analysis of the proposed A-PIAAE. We prove the convergence of the proposed A-PIAAE, and from the convergence error we can see that the prior information applied to MDs can improve the overall recovery performance.
	\item[4)] We construct various experiments to demonstrate the superiority of the proposed A-PIAAE in improving the recovery accuracy for ADs and reducing the average AoI for MDs. Experimental results show that the proposed method outperforms both the traditional CS methods and DL methods.
\end{itemize}

We organize the remainder of this paper as follows. In section \uppercase\expandafter{\romannumeral2}, we model the access problem for heterogeneous devices and propose an optimization objective that aims to simultaneously minimize the recovery error for ADs and the AoI for MDs. Section \uppercase\expandafter{\romannumeral3} introduces the optimized ARA scheme and the proposed A-PIAAE in detail. In section \uppercase\expandafter{\romannumeral4}, we analyze how the prior information improves the performance and prove the convergence of the proposed A-PIAAE and its performance gain. Experimental results in section \uppercase\expandafter{\romannumeral5} confirm the performance gain of the proposed scheme and algorithm. Conclusions are given in section \uppercase\expandafter{\romannumeral6}.

The notations in this paper are defined as follows. Boldface lowercase letters and uppercase letters represent column vectors and matrices, respectively. $(\cdot)^\top$ defines the transposition operation. $\|\cdot\|_0$ means the $\ell_0$ norm that counts the number of non-zeros, and $\|\cdot\|_1$ denotes the $\ell_1$ norm that calculates the sum of the absolute values of the vector elements. Let $\mathcal{S}$ denotes the set of indices, then $\mathbf{a}_\mathcal{S}$ and $\mathbf{A}_\mathcal{S}$ denotes the new vector composed of elements that are indexed by $\mathcal{S}$ and sub-matrices composed of the rows of matrix $\mathbf{A}$ contained in $\mathcal{S}$ respectively. $\mathbf{A}_i$ represents the $i$-th row of matrix $\mathbf{A}$.
\section{Related Work}
\subsection{Heterogeneous Services}
There have been many studies focusing on the efficient allocation of the radio time-frequency resources and preamble resources to guarantee the performance of heterogeneous services \cite{10819487,10654339,10517305,10238402,10190330,10670292}. To ensure the needs of mMTC devices and ultra-reliable low-latency communications (URLLC) devices, an intelligent preamble allocation scheme is proposed in \cite{10238402}, which uses hierarchical reinforcement learning to partition the URLLC devices exclusive preamble resource pool at the base station side and perform preamble selection within each slot at the device side. To allocate resources to different services such as enhanced mobile broadband, mMTC, URLLC, a promising scheme based on rate-splitting multiple access is studied in \cite{10190330}, which can provide a more flexible decoding order and theoretically has the largest achievable rate region than orthogonal multiple access (OMA) and non-orthogonal multiple access (NOMA). A novel  contention-based grant-free scheme in which mMTC/URLLC coexist is proposed in \cite{10670292}. MMTC devices transmit packet replicas having different preambles in various time slots, capitalizing on the temporal domain.And the URLLC devices leverage pilot mixture and power diversity to meet the stringent latency and reliability requirements. Contention resolution is achieved through successive interference cancellation. However, there is a lack of research on the coexistence of status update service devices and alarm-driven devices in mMTC, aiming to fulfill both the high information timeliness requirements of MDs and the high successful detection rate demands of ADs. This paper incorporates device-type heterogeneity into a GF-RA optimization framework. Based on the AoI, we design a prior-information-aided access mechanism and detection algorithm to achieve coordinated performance optimization for heterogeneous devices.

\subsection{Deep Unfolding Network}
To enhance AUD performance and reduce detection time, deep learning (DL) has been used to solve sparse recovery problems \cite{9903376,10786246}. The deep-unfolding methods, for example, the learned ISTA (LISTA) \cite{10.5555/3104322.3104374}, learned AMP (LAMP) \cite{7934066}, analytic learned proximal gradient
method \cite{10571841} and alternating direction method of multipliers-based network \cite{10292744}, extend traditional iterative sparse recovery algorithms into neural network (NN) architectures and have demonstrated considerable promise in enhancing accuracy and reducing complexity.

Numerous studies have focused on modifying existing deep unfolding network architectures to further improve recovery performance and enhance convergence behavior. To guarantee the convergence of the network, a framework with safeguard is proposed in \cite{Heaton_Chen_Wang_Yin_2023}. The safegurad is activated and the conventional optimization algorithm replaces the current iterative update when it performs poorly or appears to diverge. Utilizing the information distilled from the initial data recovery phase as the prior information, a prior information aided network is proposed in \cite{9605579} to enhance AUD and channel estimation. Exploring the characteristics of asynchronous access and varying data lengths, a modified LAMP that incorporates a backward propagation algorithm is introduced \cite{10304065}, which exploits the three-level sparsity resulting from sporadic activity, symbol delays, and data length diversity. In \cite{10615768}, a parameter estimation module to adapt different active ratios and noise variance based on the expectation maximization algorithm is introduced within the learned vector approximate message propagation network. A refinement module is designed in \cite{10192549} to further advance the performance of LAMP by utilizing the spatial feature caused by the correlated sparsity pattern. 
 However, most studies on GF-RA concentrate on how to reduce inter-user interference and improve the performance of user detection and data recovery, treating all users as equally important. If GF-RA is considered with AoI as the metric, the aforementioned methods may not achieve optimal system performance. This work is based on a deep unfolding network architecture, where AoI of all devices is introduced as prior knowledge to design the LISTA-AGE framework. Furthermore, an end-to-end autoencoder architecture, A-PIAAE, is proposed to enable the joint optimization of pilot design and device detection, thereby simultaneously improving detection accuracy and overall system AoI performance.


\subsection{Age-optimal Random Access}
Previous studies focusing on optimizing AoI, such as \cite{8807257,9488702,8933047}, have explored scheduling policies for multiple access channels under centralized control. However, in mMTC scenarios where devices are decentralized, implementing such policies would necessitate significant communication and coordination overhead, making them impractical. Towards designing decentralized algorithms for minimizing
age of information, references \cite{9162973,9377549,10615953,10461010,10042422,10123405} consider the slotted-ALOHA frame, and an age-based scheme is further designed in \cite{9162973,9377549,10615953} in which a device is permitted to access the channel with constant probability only when its AoI is greater than a certain threshold. The expression for steady-state expected AoI is obtained and the access probability and threshold are optimized. In \cite{9791264}, adaptive status update scheme is studied and the threshold of AoI can further be adjusted through a transmission feedback mechanism.

The age-optimal scheduling are further studies with heterogeneous traffic. A optimal scheduling policy is proposed in \cite{9365698}, aiming at minimizing the average AoI of the energy harvesting node subject to the queue stability condition of the grid-connected node. A scheduling scheme is developed in \cite{10438833} that minimizes the long-term average AoI under the delay tolerance constraint, to meet the varying transmission timeliness requirements for different arrival models. The achievable performance of the average AoI for AoI-oriented users and the average secrecy rate for throughput-oriented users is investigated in \cite{10495345} in a multiuser heterogeneous uplink wireless network. 
Nevertheless, existing literature on AoI-optimal scheduling schemes for both single-service or heterogeneous-service scenarios focus on OMA strategies. When collision occurs, the access is failed and devices need to retransmit. Such orthogonal access schemes may not perform well for massive devices.

Some studies have concentrated on the AoI optimization in NOMA scheme. Asymptotic studies are carried out in \cite{10260312} to demonstrate that the use of the simplest form of NOMA is already sufficient to reduce the AoI of OMA by more than 40\%. Considering the downlink short-packet communication via NOMA, the expected weighted sum AoI is minimized in \cite{9500896} by optimizing the power allocation. Focusing on NOMA-assisted wireless backscatter networks, a joint optimization problem encompassing backscatter coefficient configuration, node matching, and time slot allocation is solved in \cite{10722857}. However, these studies do not concern the joint optimization of AUD and channel estimation during the initial access phase, which we want to solve by designing the access scheme and detection algorithms.

To bridge these gaps, this paper integrates GF-RA with AoI optimization by introducing CS framework and prior information. A distributed access model for heterogeneous services is constructed, in which threshold and access probability parameters are systematically optimized. Furthermore, a jointly trained detection architecture is designed to address the trade-off between AoI performance and detection accuracy under massive device access scenarios.

\section{System Description}
\begin{figure}[t]
	\centerline{\includegraphics[scale=0.34]{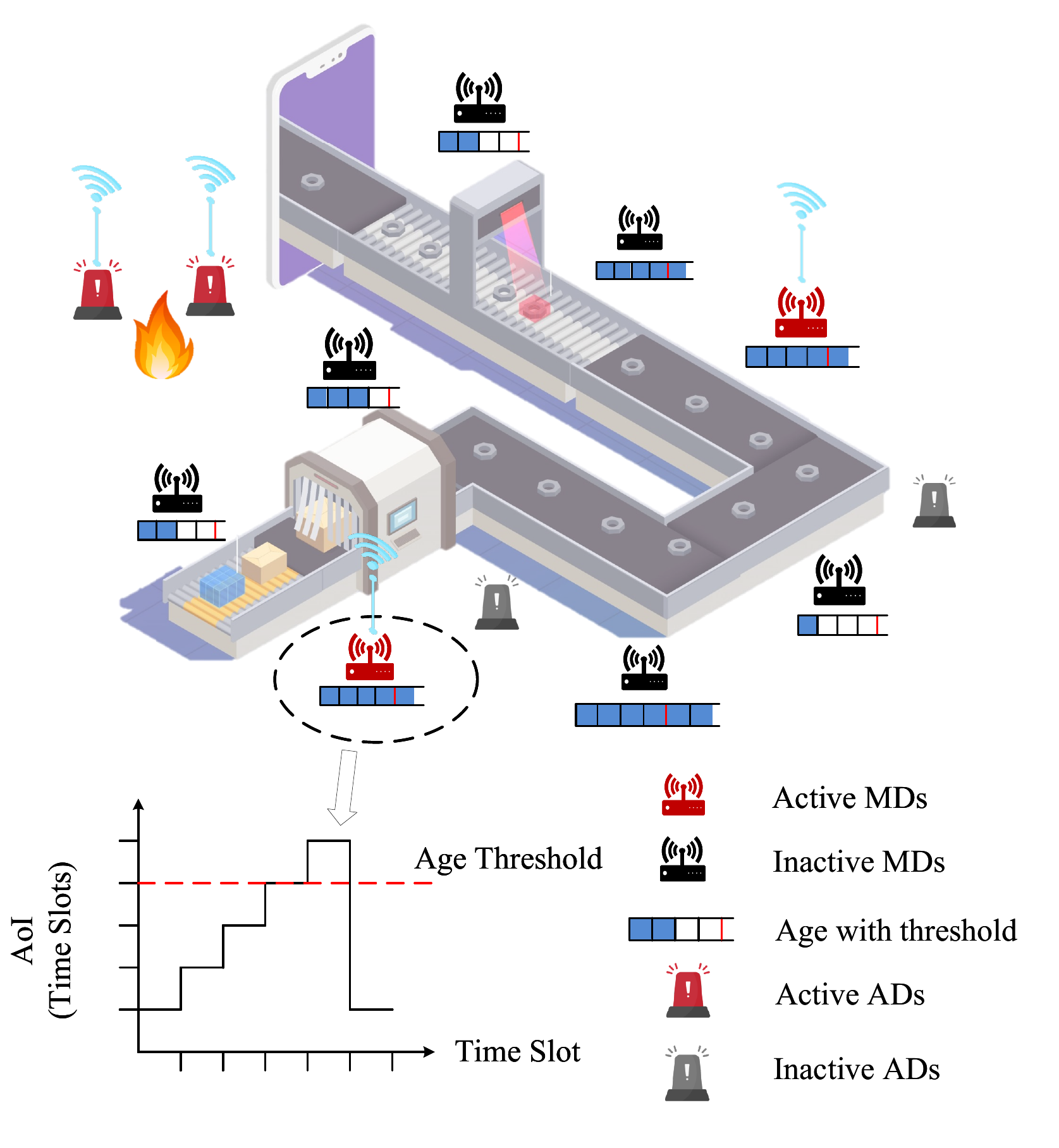}}
	\caption{An illustration of the system model. The queue under MDs represents their AoI and the red line in the queue represents the age threshold. The activation of MDs follows ARA scheme, where MDs with AoI below the threshold is inactive and become active with probability $p$ if AoI is larger than the threshold. The activation of ADs is triggered by the fire alarm.}
	\label{scenario}
\end{figure}
\subsection{System Model}
As shown in Fig. \ref{scenario}, we consider an uplink mMTC scenario where two types of devices are contained. $N$ ADs observe critical triggering information, while $K$ MDs monitor the environment and report monitoring information to the BS. These devices sporadically communicate with the serving BS. The set of ADs and MDs are denoted by $\mathcal {N}=\{1,\ldots,N\}$ and $\mathcal {K}=\{1,\ldots,K\}$, respectively. The system is time-slotted, indexed by $t$, where the duration of a time slot is the channel coherence time. In each time slot $t$, $N_t\, (N_t\ll N)$ ADs and $K_t \, (K_t\ll K)$ MDs are active. We consider a CS-based GF-RA scheme where each device is assigned a unique non-orthogonal pilot. The pilot matrixes of ADs and MDs are defined
as $\mathbf{A}=\left[\mathbf{a}_1, \mathbf{a}_2, \ldots, \mathbf{a}_{N}\right] \in \mathbb{C}^{M \times N}$ and $\mathbf{B}=\left[\mathbf{b}_1, \mathbf{b}_2, \ldots, \mathbf{b}_{K}\right] \in \mathbb{C}^{M \times K}$, where $\mathbf{a}_n$ and $\mathbf{b}_k$ denote the pilot sequence assigned to AD $n$ and MD $k$ respectively, $M$ denotes the length of pilot and $\mathbb{C}$ denotes the complex space. In each time slot, every active MD that decides to transmit generates a fresh status update packet just before transmission, which is known as the ``generate-at-will'' model \cite{8000687}. We neglect the queuing time of data in the network, assuming that in each time slot, the data from all active MDs arrive at the BS simultaneously. When a triggering event occurs in the same time slot, the corresponding ADs are activated. Then, these active devices directly send their own reserved pilot and the packet data to the BS without scheduling.

In time slot $t$, let $\mathcal{N}_t \subset\{1,2, \ldots, N_t\}$ denote the index of active ADs and $x_{n,t}$ denote the channel gain of AD $n$ at time $t$. Similarly,  $\mathcal{K}_t \subset\{1,2, \ldots, K_t\}$ denotes the index of active MDs and $g_{k,t}$  denotes the channel gain of MD $k$ at time $t$. Following the idea in \cite{7331960}, we apply our model in the orthogonal frequency division multiple (OFDM) system, where all active devices transmit their pilots and data on the same physical resource blocks, leading to interference between them. We assume that all transmitted
packages are  synchronized at the BS for simplicity, then the received superimposed pilot signal $\mathbf{y}_t \in \mathbb{C}^{M\times 1}$ at the BS is given by
\begin{equation}\label{cs}
	\begin{aligned}
		\mathbf{y}_t &= \sum\limits_{n \in \mathcal{N} } \alpha_{n,t}{{x_{n,t}}} {\mathbf{a}_n} + \sum\limits_{k \in \mathcal{K} } \beta_{k,t}{{g_{k,t}}} {\mathbf{b}_k} + \mathbf{n}_t\\ 
		& = \mathbf{A}\mathbf{x}_t + \mathbf{Bg}_t+ \mathbf{n}_t\\
		& = \mathbf{Ph}_t + \mathbf{n}_t,
	\end{aligned}
\end{equation}
where $\alpha_{n,t}$ is an indicator with $\alpha_{n,t}=1$ if AD $n$ is active in time slot $t$, otherwise $\alpha_{n,t}=0$. Similarly, $\beta_{k,t} = 1$ if MD $k$ is active in time slot $t$, otherwise $\beta_{k,t}=0$. $\mathbf{x}_t=\left[\alpha_{1,t}x_{1,t}, \alpha_{2,t}x_{2,t}, \ldots, \alpha_{N,t}x_{N,t}\right]^\top \in \mathbb{C}^{{N} \times 1}$  and $\mathbf{g}_t=\left[\beta_{1,t}g_{1,t}, \beta_{2,t}g_{2,t}, \ldots, \beta_{K,t}g_{K,t}\right]^\top \in \mathbb{C}^{{K} \times 1}$ are sparse vectors with $N_t$ and $K_t$ non-zero elements, respectively, due to the sporadic communication characteristic. Let $\mathbf{P}\triangleq [\mathbf{A},\mathbf{B}] \in \mathbb{C}^{M \times S} $ and $\mathbf{h}_t \triangleq [\mathbf{x}_t,\mathbf{g}_t]^\top \in \mathbb{C}^{S \times 1}$, where $S = N + K$ denotes the total devices of the system. We adopt the block-fading channel model, where in each time slot the channel follows independent quasi-static flat-fading. We assume the channel of any device $s$ in the system is $h_{s,t} = l_{s,t}z_{s,t}$, where $l_{s,t}$ is the path-loss determined by device location in time slot $t$ and $z_{s,t}$ follows $\mathcal{C N}(0,1)$.  $\mathbf{n}_t \in \mathbb{C}^{M \times 1}$ denotes the additive white Gaussian noise. Accordingly, the channel estimation and AUD is a sparse recovery problem that can be solved via CS algorithms.

\subsection{Device Activity}
We first define the average AoI to measure the information freshness of MDs. AoI of MD $k$ at time slot $t$, $\Delta_k(t), t=1,2,...,T$, is defined as the number of slots that have elapsed since the freshest status update packet of this MD was generated to date and successfully received by the BS. $\Delta_k(t)$ is equal to the number of slots of last time slot plus one since the most recent successful detection of MD $k$. In the case of a successful detection, the BS sets the AoI of successfully detected MDs to one, which is because we assume that packets are transmitted at the beginning of the time slots and received by BS at the end of the same time slots. The corresponding successfully detected MDs will receive the acknowledgment message from BS and reset their AoI to one. Accordingly, the evolution of the instantaneous AoI $\Delta_k(t)$ can be expressed as
\begin{equation}\label{aaoi}
	\Delta_k(t+1)=\left\{\begin{array}{l}
		1,    \quad\qquad\qquad \text { if } \ I_k(t)=1, \\
		\Delta_k(t)+1, \quad \text { otherwise },
	\end{array}\right.
\end{equation}
where $I_k(t)$ is the indicator of successful detection. Particularly, $I_k(t)=1$ means that the active MD $k$ is successfully detected, and $I_k(t)=0$ means MD $k$ is inactive or unsuccessfully detected.
Based on the AoI evolution, the average AoI $\bar{\Delta}_k$ for each MD $k$ can be defined as 
\begin{equation}
	\bar{\Delta}_k=\lim _{T \rightarrow \infty} \frac{1}{T} \sum_{t=1}^T \Delta_k(t).
\end{equation}

We adopt ARA scheme for MDs, which is an age-aware extension of slotted-ALOHA. ARA scheme schedules MDs in a decentralized way, in which the decision of each MD at time slot $t$ is determined by its AoI at the beginning of this time slot. That is, if the instantaneous AoI is larger than the threshold ${\delta}$, MD becomes active with a fixed probability $p$. Otherwise, MD is inactive. 

Then, define the average AoI of MD ${\bar \Delta }(\delta,p)$ as a function of ${\delta}$ and $p$, which is derived in the subsequent section. This consideration is motivated by the intuition that if the AoI of a device is larger than the threshold, the BS has not updated the information of that device in a long time. This kind of device has a more urgent need for access and should be granted a higher access probability, while other devices access the channel with a lower probability (in this paper the probability is equal to 0) to mitigate interference. In other words, this scheme reduces the average AoI of MDs by allowing MDs that contribute more significantly to reducing the average AoI to access the channel with higher probability.

For ADs, their activation is event-triggered. For example, when a fire occurs, the corresponding ADs become active and send relevant information to the BS, requiring a high successful detection rate.

\subsection{Receiver Description}

\begin{figure}[t]
	\centerline{\includegraphics[scale=0.38]{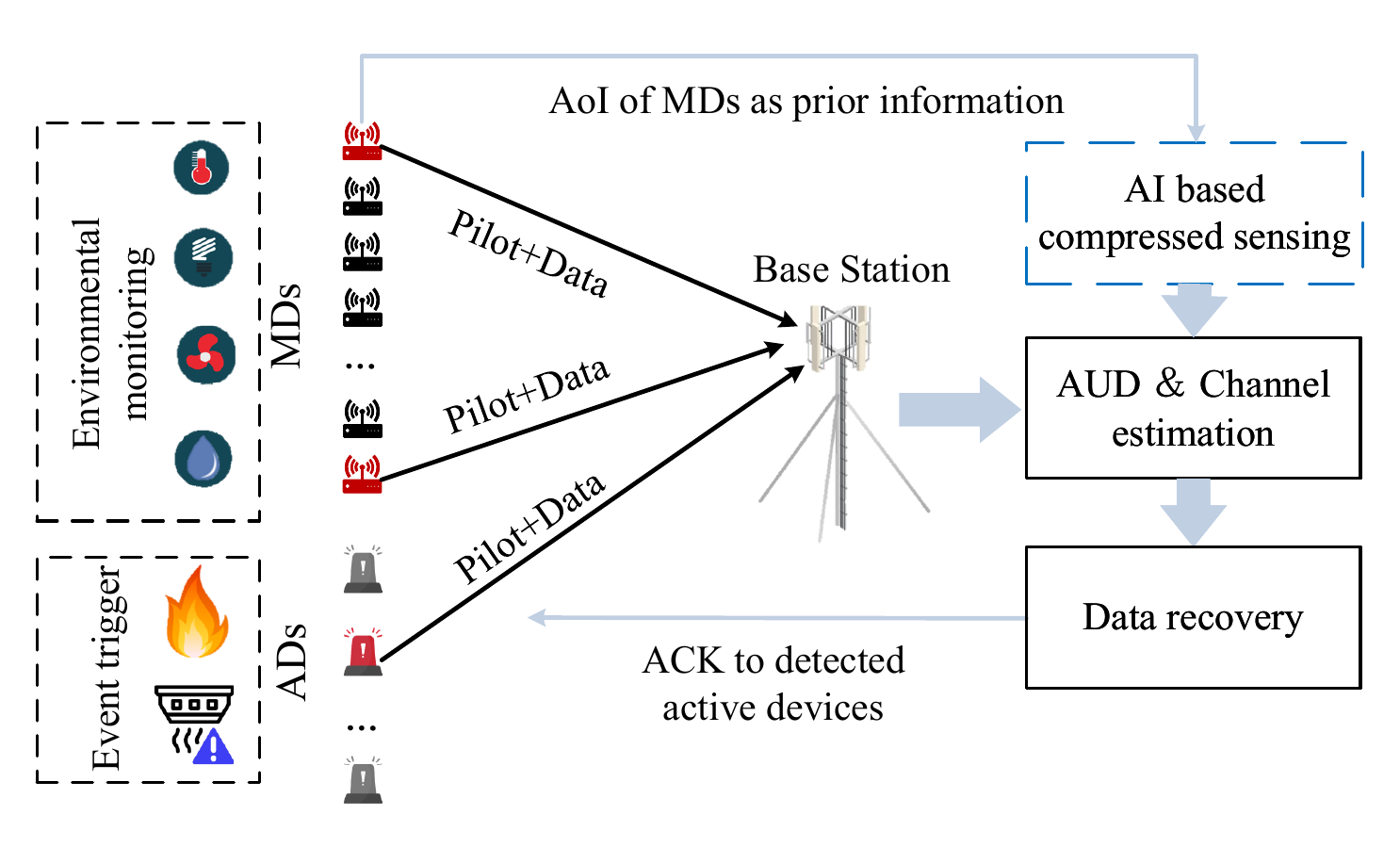}}
	\caption{An illustration of the access process of devices with heterogeneous traffic.}
	\label{process}
\end{figure}

The overall access process of devices with heterogeneous traffic is illustrated in Fig. \ref{process}. After receiving $\mathbf{y}_t$, the BS first performs AUD and channel estimation using AI-based CS methods. Then the estimated index set of active devices $\hat{\mathcal{K}}_t$, $\hat{\mathcal{N}}_t$ and estimated channel $\hat {\mathbf{h}}_t$ can be obtained. Following the identification of the active device index set, data recovery can be performed efficiently using the least squares algorithm. AUD and channel estimation are the most critical steps in the entire data recovery process.

Let $S_t$ denote the total active devices in slot $t$ with $S_t = N_t+K_t$. Then $\mathbf{h}_t$ is a sparse vector with $\|\mathbf{h}_t\|_0=S_t$. Known $\mathbf{y}_t$ and the pilot matrix $\mathbf{P}$, the estimation of channel vector $\mathbf{h}_t$ can be formulated into an $\ell_0$ norm minimization problem using CS methods
\begin{equation}\label{l0}
	\min _{\mathbf{h}_t}\|\mathbf{h}_t\|_0 \quad \text { s.t. }\left\|\mathbf{y_t}-\mathbf{P h}_t\right\|_2^2 \leq \varepsilon,
\end{equation}
where $\varepsilon>0$ depends on the noise level. As this $\ell_0$ norm minimization problem is NP-hard, a common way to solve this kind of problem is to replace it with $\ell_1$ norm minimization
\begin{equation}\label{l1}
	\min _{\mathbf{h}_t}\|\mathbf{h}_t\|_1 \quad \text { s.t. }\left\|\mathbf{y}_t-\mathbf{P h}_t\right\|_2^2 \leq \varepsilon.
\end{equation}
With the estimated channel $\hat {\mathbf{h}}_t\in \mathbb{C}^{S \times 1}$, we can identify which pilots are selected, indicating the active devices.

In GF-RA, the BS first performs AUD and channel estimation, followed by data recovery based on the estimated channels. However, if the channel estimation error is significant, it can lead to the failure of the subsequent data recovery process. In this paper, considering that the performance of data recovery highly relies on the AUD and channel estimation in \eqref{cs}, we assume a device is detected successfully if the channel estimation error is smaller than a threshold $\tau$ and the data of device can always be recovered accurately \cite{9605579},\textcolor{blue}{\cite{8961111}}. So the indicator in \eqref{aaoi} for determining whether the device $k$ is successfully detected can be expressed as 
\begin{equation}\label{criteria}
	{I_k }(t) = \left\{ {\begin{array}{*{20}{l}}
			{1,{\rm{     \quad if }}\;k \in {\cal K}_t{\rm{ \, and \, }}\left| {g_{k,t} - \hat g_{k,t}} \right| \le \tau,}\\
			{0,\quad {\rm{ otherwise }},}
	\end{array}} \right.
\end{equation}
where $\hat g_{k,t}$ is the estimated channel of MD $k$. As shown in \eqref{criteria}, the AUD criterion exhibits robustness against a certain degree of channel estimation imperfection. Even when the channel state information is imperfect, active devices can still be accurately detected within a bounded error margin.

As mentioned before, if one MD is inactive or fails to access the BS, its AoI is increased by one based on its value from the previous time slot. If the MD is successfully detected, its AoI is reset to one. The evolution of AoI can be modeled as a Discrete-Time Markov Chain according to \cite{9162973}, from which the stationary probability distribution over all possible AoI values can be derived. Then, the average AoI for each MD can be derived as
	\begin{equation}\label{average AoI}
		\bar{\Delta}(\delta,p,q)=\sum_{c=1}^{\infty} c \pi_c=\frac{\delta}{2}+\frac{1}{p q}-\frac{\delta}{2(\delta p q+1-p q)},
	\end{equation}
	where $\pi_c$ $(c=1,2, \cdots)$ denotes the stationary probability of the instantaneous AoI being $c$, and $q$ denotes the probability of successful detection which can be derived from $\delta$ and $p$. So $\bar{\Delta}(\delta,p,q)$ can be reformulated as $\bar{\Delta}(\delta,p)$

For the optimization object, this paper investigates two types of devices with distinct access requirements and performance metrics. ADs require a high successful detection rate, while MDs prioritize delivering the freshest possible information. Instead of directly solving \eqref{l1}, our goal is to minimize the average AoI of MDs and recover the sparse $\mathbf{h}_t$ simultaneously, i.e.,
\begin{subequations}
	\begin{align}
		\min_{\delta,p,\mathbf{P}}& \ {\bar \Delta }(\delta,p) + \|h_{\{\mathbf{P}\}}(\mathbf{y}_t)\|_1  \\
		\text { s.t. }\left\|\mathbf{y}_t-\mathbf{P} h_{\{\mathbf{P}\}}(\mathbf{y}_t)\right\|_2^2 \leq \varepsilon, \label{c1}\\
		& \delta \in \mathbb{Z}_+,\label{c2}\\
		& p \in [0,1],\label{c3}
	\end{align}
\end{subequations}
where $ \mathbf{h}_t = h_{\{\mathbf{P}\}}(\mathbf{y}_t)$ indicates the NN function of recovery $\mathbf{h}_t$ parameterized by the pilot matrix $\mathbf{P}$. $\mathbb{Z}_+$ denotes positive integer. Constraint (\ref{c1}) ensures that the recovery error is less than $\varepsilon$. Constraint (\ref{c2}) indicates that the age threshold $\delta$ is a positive integer. Constraint (\ref{c3}) indicates that the value of access probability $p$ lies in the range of 0-1. The optimized parameters are the age threshold $\delta$, the access probability $p$ and the parameter matrix $\mathbf{P}$ in NN for recovering $\hat{\mathbf{h}}_t$.
\section{Random access scheme and active user detection}
In this section, we first optimize the ARA scheme by optimizing the access parameters of MDs. Then we propose an autoencoder A-PIAAE to jointly optimize the pilots of devices and recover the channel vector $\mathbf{h}$. Furthermore, based on the proposed ARA scheme, we design an age-based decoder LISTA-AGE in A-PIAAE by exploiting AoI as the prior information.

\subsection{Age-Based Random Access Scheme }\label{AA}
The average AoI achieved by the ARA scheme has been demonstrated to be around half the minimum AoI achieved by regular slotted-ALOHA in \cite{9163053,9377549}. However, these analyses are based on collision channel models, where a collision results in access failure if two or more devices simultaneously attempt to access the channel. In the mMTC scenario, characterized by massive devices access, there is a significant likelihood that multiple devices will access the same OFDM resources, leading to resource conflicts. To address this challenge, we propose to leverage CS techniques, enabling the BS to successfully detect active devices despite such conflicts. In this subsection, we analyze the successful detection rate based on the CS theory and derive the optimal access parameters.

\begin{figure}[t]
	\centerline{\includegraphics[scale=0.7]{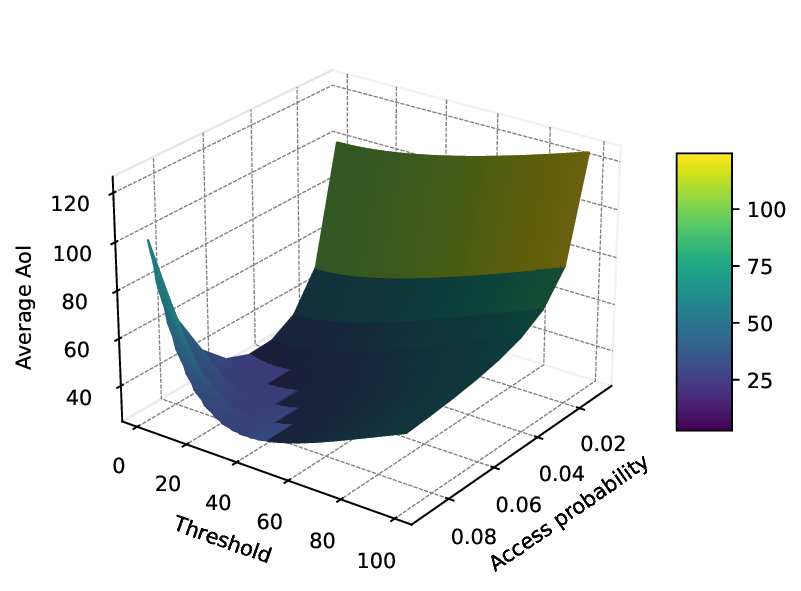}}
	\caption{A two-dimensional schematic diagram of the AoI function with parameters $p$ and $\delta$ is presented. Since the AoI becomes significantly large when $p > 0.1$, the range of $p$ is restricted to 0 - 0.1. The diagram indicates that the AoI has a minimum value, which can be determined using two-dimensional search.}
	\label{3D}
\end{figure}
In \eqref{average AoI}, the probability of successful detection $q$ is related to ${\delta}$ and $p$, and is derived according to the probability of collision between devices in \cite{9162973}. Nevertheless, it is not easy to derive in GR-RA. Using CS theory, we derive the expression for $q$, then jointly optimize the threshold parameter ${\delta}$ and $p$ to minimize the average AoI of MDs.


A useful rule of thumb is discovered in \cite{0901.3403} that the oversampling factor $\frac{M}{ S_t}$ for $\ell_1$ reconstruction satisfies:
\begin{equation}
		\frac{M}{ S_t} = \log _2\left(1+\frac{S}{S_t}\right).
\end{equation}
Since a larger $M$ means that the more samples are taken in the CS problem, the easier the problem is to be solved, we have that the CS problem \eqref{l1} can be solved successfully when the following condition holds
\begin{equation}
		M \geqslant S_t \log _2\left(1+\frac{S}{S_t}\right)\buildrel \Delta \over = \phi({S_t}).
\end{equation}
For fixed $M$ and $S$, the derivative of $\phi(S_t)$ with respect to $S_t$ is given by
\begin{equation}
		\frac{{\partial \phi(S_t)}}{{\partial S_t}} = {\log _2}\left(1 + \frac{S}{S_t}\right) - \frac{S}{({S_t + S})\ln2}.
\end{equation}Since $\mathbf{h}_t$ is a sparse vector, sparsity $S_t$ is a positive integer far less than $S$,
 $\frac{{\partial \phi(S_t)}}{{\partial S_t}}>0$ and $\phi(S_t)$ is a monotonically increasing function. Given $M$ and $S$, we can derive the maximum sparsity $S_{\max}$, which means the CS problem can be solved if the sparsity is less than $S_{\max}$. And the maximum number of active MDs is $S_{\max}-N_t$. 

 \begin{algorithm} [!t] 
	\caption{The two-dimension search algorithm}  \label{search algorithm}
	\begin{algorithmic}[1]
		\State Initialization:
		$p_{min} = 0$, $p_{max} = 1$, $p_{step} = 0.01$, $\delta_{min} = 1$, $\delta_{max} = 100$, $\delta_{step} = 1$. $\bar{\Delta}_{min} = +\infty$
		\State Generate the value set of $p$ and $\delta$: $ \{p_{min},p_{min}+p_{step},...,p_{max}\}$ and $ \{\delta_{min},\delta_{min}+\delta_{step},...,\delta_{max}\}$.
		\For {each $(p,\delta)$ pairs in the value set}
		\State calculate the average AoI $\bar{\Delta}$ according to \eqref{average AoI} and \eqref{q}.
		\If{$\bar{\Delta} < \bar{\Delta}_{min}$}
		\State Update $\bar{\Delta}_{min} = \bar{\Delta}$.
		\State Update the optimal access parameters pair $(p_m,\delta_m)$.
		\EndIf
		\EndFor 
		\State Return $p_m,\delta_m$
	\end{algorithmic}
\end{algorithm}

Next, we calculate the probability that the number of active MDs is less than or equal to $S_{\max}-N_t$, which is the successful detection rate. Since in each time slot, the activity of each MD is independent, the number of active MDs at a given slot obeys a binomial distribution. MDs with an AoI less than the threshold $\delta$ remain inactive and are excluded from the count of active MDs. We simplify the assumption that the AoI of every MD follows a uniform distribution over the range from 1 to $a_{\max}$, where $a_{\max}$ is the maximum AoI during the AoI evolution. Via the simulations in section IV, we will simulate the real AoI evolution process and verify that the access parameters determined under this assumption closely resemble the optimal parameters. The number of active MDs $K_t$ in slot $t$ follows the binomial distribution
	\begin{equation}\label{binomial distribution}
		K_t \sim B\left( {\frac{{a_{\max} - \delta }}{{a_{\max}}}K,p} \right).
	\end{equation} 
Then, the probability of $K_t$ less than or equal to $S_{\max}-N_t$, which is the successful detection rate of MDs, can be calculated using the theorem 1, as show below.
\newtheorem{thm}{Theorem}
\begin{thm}[The successful detection rate]
	The successful detection rate can be approximated as
	\begin{equation}\label{q}
			\begin{aligned}
				q &=P\left( {K_t \le {S_{\max }-N_t}} \right)\\
				&=\sum\limits_{k=1}^{{S_{\max }-N_t}} {C_{\frac{{a_{\max} - \delta }}{{a_{\max}}}K}^k{p^k}{{(1 - p)}^{\frac{{a_{\max} - \delta }}{{a_{\max}}}K - k}}},
			\end{aligned}
		\end{equation}
		where $C_{\frac{{a_{\max} - \delta }}{{a_{\max}}}K}^k$ is the binomial coefficient.
\end{thm}

Substituting \eqref{q} into \eqref{average AoI}, the average AoI is a function of $p$ and ${\delta}$. As shown in Fig. \ref{3D}, the average AoI has a minimum value with respect to $p$ and ${\delta}$. According to Algorithm 1, given $M$ and $S$, we can determine the optimal $p_m$ and ${\delta}_m$ via a two-dimension search to minimize the average AoI of MDs. 
And it should be clarified that the optimal access parameters only depend on the pilot length $M$, the total number of devices $S$ and the number of active ADs $N_t$, and all the MDs are applied the same optimal parameters in ARA scheme.

\subsection{Autoencoder Design}
In the previous subsection, we analyze the ARA scheme and optimize its access parameters. The decentralized implementation of the optimized ARA scheme leads to the attainment of minimum average AoI for MDs. From \eqref{average AoI}, it is evident that the successful detection rate $q$ is a critical factor influencing the average AoI. An increase in $q$ (attributable to improvements in the detection algorithm, while $p$ and $\delta$ remain constant) results in a reduction in the average AoI. Similarly, ADs also rely on a high successful detection rate to transmit urgent messages effectively. In this subsection, we apply the DL technique to improve the successful detection rate for heterogeneous traffic devices.

As shown in Fig. \ref{figae}, we propose a data-driven autoencoder architecture, namely A-PIAAE, to jointly design the pilot matrix and recover the sparse signal for AUD.
In \cite{9174792,9754266}, the pilot matrix is jointly optimized with NN, using the DL autoencoder, which can effectively exploit the properties of the sparsity patterns. To make full use of the pilot resources, the two kinds devices are allocated non-orthogonal pilots. While this pilot resource allocation method enables support for a larger number of devices, the increase in the number of non-orthogonal pilots will lead to greater interference among them, complicating AUD via CS algorithms. The joint optimization of the pilot matrix using an autoencoder enhances the orthogonality among non-orthogonal pilots, effectively reducing pilot interference and thereby improving overall system performance.

For simplicity, the time index $t$ is omitted in the following expressions. As shown in Fig. \ref{figae}, the deep unfolding autoencoder consists of a linear encoder and a nonlinear decoder. The enconder is represented by
\begin{equation}
	g(\mathbf{h}):=\mathbf{P} \mathbf{h}+\mathbf{n},
\end{equation}
which mimics the pilot transmission in wireless channel. $\mathbf{P}$ is the concatenate matrix that contains the pilot matrices $\mathbf{A}$ and $\mathbf{B}$. Furthermore, we normalize the columns of the pilot matrix to constrain the energy of all pilot sequences to 1 in the power normalization module.
The nonlinear decoder is defined by
\begin{equation}\label{decoderp}
	f(\mathbf{y}):=\underset{\mathbf{h}}{\operatorname{argmin}}\|\mathbf{h}\|_1 \quad \text { s.t. }  \left\|\mathbf{y}-\mathbf{P h}\right\|_2^2 \leq \varepsilon,
\end{equation}
which we will elaborate in the next subsection.

\begin{figure*}[t]
	\centerline{\includegraphics[scale=0.4]{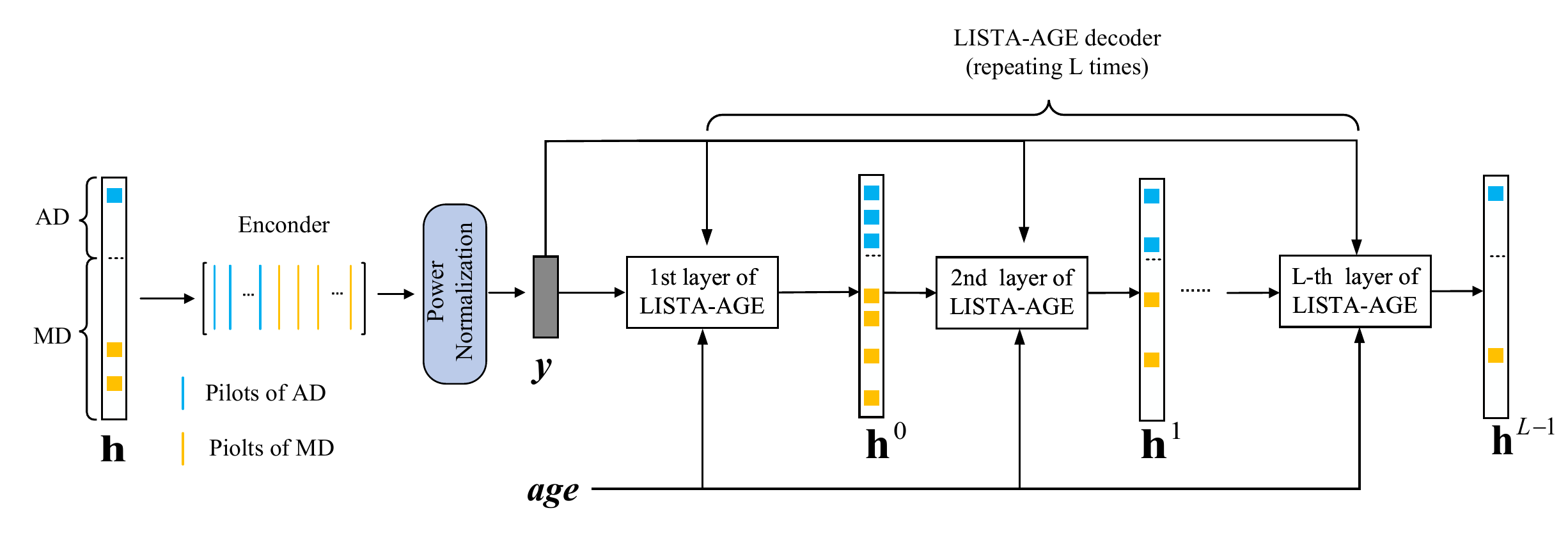}}
	\captionsetup{justification=centering}
	\caption{An illustration of A-PIAAE.}
	\label{figae}
\end{figure*}

\subsection{Age Aided Decoder}
NNs for solving \eqref{decoderp} already exist, such as the famous LISTA and its siblings. In conjunction with the ARA scheme, we leverage the prior knowledge that certain devices are guaranteed to be inactive. Incorporating this prior information, we modify the LISTA structure and design a novel decoder, LISTA-AGE, to enhance AUD performance.

We first introduce ISTA and LISTA, which are easier to theoretically analyze, followed by a discussion on how to incorporate prior information into the network and modify LISTA accordingly. Each iteration of ISTA has the form 
\begin{equation}\label{ista}
	{\mathbf{h}^{l + 1}} = {\eta _\mathrm{st}}(\omega {\mathbf{P}^\top}\mathbf{y} + ({\mathbf{I}_S} - \omega {\mathbf{P}^\top}\mathbf{P}){\mathbf{h}^l};\theta ),
\end{equation}
where ${\eta _\mathrm{st}}(\mathbf{x},\lambda)=\operatorname{sign}(\mathbf{x}) \max \{|\mathbf{x}|-\theta, 0\}$ denotes the soft-thresholding function, $\operatorname{sign}(\mathbf{x})$ is the sign function and $\theta$ is a pre-defined parameter. The purpose of this parameter is to shrink the values of elements less than it to zero to achieve a forced sparsity effect. $\mathbf{I}_S$ denotes the unit matrix of dimension $S$. The step size $\omega$ is usually taken to be the largest eigenvalue of $\mathbf{P}^\top\mathbf{P}$ empirically. However, the recovery performance of traditional CS methods highly relies on the certain conditions of the sensing matrix $\mathbf{P}$, such as the restricted isometry property (RIP) \cite{CRMATH_2008__346_9-10_589_0}. In mMTC scenarios, communication resource constraints limit pilot length, and the potentially large number of active devices further impacts the conditioning of the sensing matrix, increasing the likelihood of failing to meet the RIP constraint.
Deep unfolding method LISTA uses NN layers to replace $\omega {\mathbf{P}^\top}$ and ${\mathbf{I}_S} - \omega {\mathbf{P}^\top}\mathbf{P}$ in \eqref{ista}, i.e.,
\begin{equation}\label{lista}
	{\mathbf{h}^{l + 1}} = {\eta _\mathrm{st}}({\mathbf{W_1}\mathbf{y} + \mathbf{W_2}{\mathbf{h}^l};\theta )},
\end{equation}
and learns the parameter matrices $\mathbf{W_1}$, $\mathbf{W_2}$ and the threshold $\theta$.

We assume that the BS will update and record the AoI of all MDs after receiving the superimposed signal and detecting the active devices. Owing to the ARA scheme, the BS is aware that MDs with AoI below the threshold are guaranteed to be inactive. We exploit this prior information into the network and modify the soft-thresholding function in LISTA. The age information of MDs goes into a threshold function to indicate whether the $i$-element of $\mathbf{g}$ is determined to be zero, which is defined as
\begin{equation}
	{{\bf{\gamma}}_M}[i] = \left\{ {\begin{array}{*{20}{c}}
			0\\
			1
	\end{array}} \right.\begin{array}{*{20}{c}}
		{{\rm{  }}\mathbf{age}[i] \le \delta },\\
		{{\rm{  }}\mathbf{age}[i] > \delta },
	\end{array}
\end{equation}
where $\mathbf{\gamma}_M$ and $\mathbf{age}$ denote the indicator and age information vectors of all MDs, respectively, with dimension of $K$. ${\bf{\gamma}}_M[i]$ and $\mathbf{age}[i]$ are the $i$-th element of $\mathbf{\gamma}$ and $\mathbf{age}$, respectively. $\delta$ is the age threshold that we optimize in the previous subsection. Then, we concatenate an all-one vector of dimension $N$ to $\mathbf{\gamma}_M$ which indicates that we have no prior information about the zero elements of $\mathbf{x}$. Defined as ${\bf{\gamma }}=[\mathbf{1}^\top,{\bf{\gamma }}_M^\top]^\top$, we use this indicator to perform element-wise operations in the thresholding function. According to the indicator, the age aided thresholding operator directly sets the elements in $\mathbf{h}$, whose corresponding AoI is below the threshold $\delta$, to 0. It is defined as 
\begin{equation}
		\eta_{\mathrm{st}}(\mathbf{h}[i];{\bf{\gamma }},\theta^l)= \begin{cases}\ 0 & \mathbf{\gamma}[i]=0, \\ \mathbf{h}[i]-\theta^l & \mathbf{h}[i]>\theta^l, \mathbf{\gamma}[i]=1, \\ \mathbf{h}[i]+\theta^l & \mathbf{h}[i]<-\theta^l, \mathbf{\gamma}[i]=1, \\ 0 & -\theta^l \leq \mathbf{h}[i] \leq \theta^l, \mathbf{\gamma}[i]=1,\end{cases}
\end{equation}where $\mathbf{h}[i]$ is the $i$-th element of $\mathbf{h}$ and ${\theta ^l}$ is the thresholding parameter of the $l$-th layer. 

As shown in Fig. \ref{one_layer}, the $l$-th layer of the decoder is defined by the following iterative of the form
\begin{equation}\label{decoder}
	{\mathbf{h}^{l + 1}} = \eta _{\mathrm{st}}(\omega{\mathbf{P}^\top}\mathbf{y} + (\mathbf{I} - \omega{\mathbf{P}^\top}\mathbf{P}){\mathbf{h}^l};\mathbf{\gamma},\theta^l).
\end{equation}
We treat the pilot matrix $\mathbf{P}$ as the weight used within the decoder, so we can parameterize the autoencoder with the matrix $\mathbf{P}$ as the tied weight of both the encoder and decoder.

\subsection{Training Procedure}


\begin{figure}[t]
	\centerline{\includegraphics[scale=0.5]{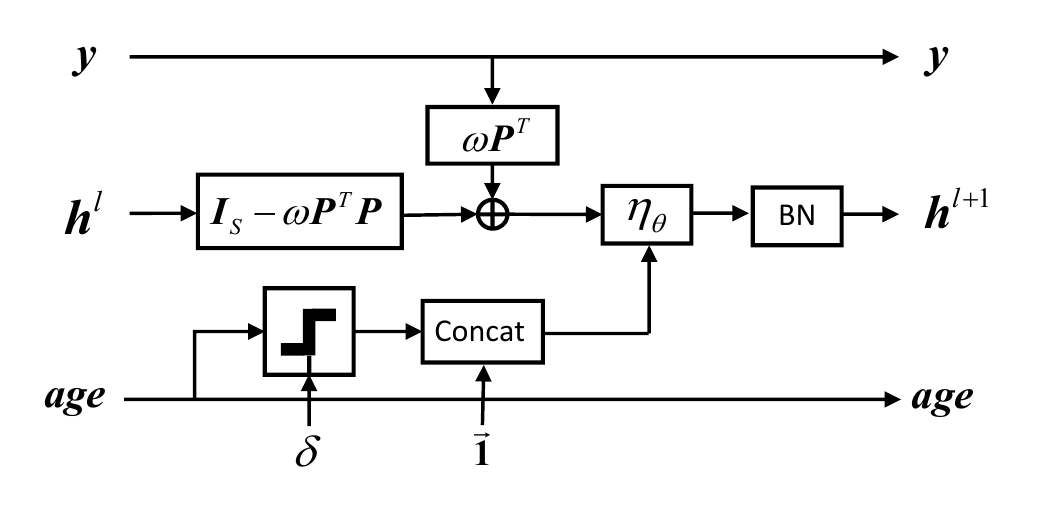}}
	\caption{An illustration of the $l$-th layer in LISTA-AGE.}
	\label{one_layer}
\end{figure}
We first generate the ground truth $\mathbf{h}_q^*$ in dataset, in which $q$ denotes the $q$-th data in the batch set. Then $\mathbf{h}_q^*$ is sent to the encoder which mimics the pilot transmission in wireless channel and $\mathbf{y}$ is got. The trainable parameters in encoder is the pilot matrix $\mathbf{P}$. Then, $\mathbf{y}$ is input into the LISTA-AGE decoder and the output is the estimated $\hat{\mathbf{h}}_q$. Since the decoder is parameterized by $\mathbf{P}$, the pilot matrix $\mathbf{P}$ is the trainable parameters in the A-PIAAE.
And it is noticed that the parameter $\omega$ in \eqref{decoder} closely relates to the corresponding pilot matrix $\mathbf{P}$, so $\omega$ must be trainable as $\mathbf{P}$ changes when the autoencoder begins training. Thus, the trainable parameters in $\Theta$ contain the matrix ${\mathbf{P}}$, $\omega$ and threshold $\left\{\theta^l\right\}_{l=1}^L$ of $L$ layers. The A-PIAAE learns the pilot matrix by using gradient-based optimization algorithm such as stochastic gradient descent to minimize the loss function defined by 
\begin{equation}\label{loss}
	\mathcal{L}(\Theta)=\sum_{q=1}^{Q} \left\| {{{\hat{\mathbf{h}}}_q} - \mathbf{h}_q^*} \right\|_2^2=\sum_{q=1}^{Q}\left\| {f(g(\mathbf{h}_q^*)) - \mathbf{h}_q^*} \right\|_2^2,
\end{equation}
where $Q$ denotes the number of training data in each batch set. The proposed A-PIAAE is trained via adaptive moment estimation (ADAM) \cite{kingma2014adam} to minimize the loss function defined in \eqref{loss}. The training procedures are given in Algorithm 2.

\section{Theoretical Analysis}
In the previous section, we design a specialized A-PIAAE and incorporate the prior information from the ARA scheme into the decoder. In CS problems, one of the primary challenges is determining the support corresponding to active devices. Knowledge of certain zero-element positions significantly reduces the complexity of identifying the support corresponding to active devices, thereby enhancing recovery performance. Rather than treating the NN as a black box, we further explore the performance gain of the proposed method in this section. Specifically, we analyze the convergence of the A-PIAAE and provide a theoretical proof of the performance gains achieved through the use of prior information.

In traditional CS algorithms, the matrices in iterations are predefined and require numerous iterations to achieve optimal performance. In contrast, LISTA learns those matrices from extensive training data and can find better matrices that lead to improved performance within fewer iterations. Several studies have proved that if the parameter matrices satisfy certain specific properties, LISTA can converge linearly with bounded errors. In this paper, we demonstrate that with the prior information used in the proposed A-PIAAE network, the constraints on the parameters are relaxed, resulting in a smaller error bound.

\newtheorem{definition}{Definition}
\begin{definition}[Basic definition]
For a noisy measurement $\mathbf{y} =\mathbf{P} \mathbf{h}^*+\mathbf{n}$, the signal $\mathbf{h}^*$ and the observation noise $\mathbf{n}$ are sampled from the following set:
	\begin{equation}
		\begin{aligned}
			\left( {{\mathbf{h}^*},\mathbf{n} } \right) &\in {\cal H}(B,s,\sigma )
			\\& \buildrel \Delta \over = \left\{ {\left( {{\mathbf{h}^*},\mathbf{n} } \right)\left| {\left| {\mathbf{h}_i^*} \right| \le B,\forall i,{{\left\| {{\mathbf{h}^*}} \right\|}_0} \le s,{{\left\| \mathbf{n}  \right\|}_1} \le \sigma } \right.} \right\}.
		\end{aligned}
	\end{equation}
In other words, $\mathbf{h}^*$ is bounded and s-sparse ($s \geq 2$), and $\mathbf{n}$ is bounded.
	
Let $\Lambda = \left\{ {i\left| {\mathbf{age}[i] \le \delta } \right.} \right\}$ which means when $i \in \Lambda$, $\mathbf{h}[i]=0$. $\Lambda$ is the set of zero element positions that we get due to the AoI prior information. 
\end{definition}

 \begin{algorithm} [!t] 
	\caption{The training procedure of A-PIAAE}  \label{algorithm procedure}
	\begin{algorithmic}[0]
		\State \textbf{Training Data:} pairs of inputs $\left\{\mathbf{h}_q,\mathbf{age}_q\right\}_{q=1}^{Q}$.
		\State \textbf{Initialization:} $\left\{\theta^l\right\}_{l=1}^L=0.1$.
		\State \textbf{Goal:} learn the parameters of A-PIAAE  $\Theta=\left\{\mathbf{P},\omega,\left\{\theta^l\right\}_{l=1}^L \right\}$.
		\For {each batch training pairs $\left\{\mathbf{h}_q,\mathbf{age}_q\right\}_{q=1}^{Q}$ }
		\State calculate the gradient of $\Theta$ according to the loss
		\State update $\Theta$ via ADAM
		\EndFor
	\end{algorithmic}
\end{algorithm}

\begin{definition}[Mutual Coherence]
Since that each column of the pilot matrix $\mathbf{P}$ is normalized for power normalization, the diagonal elements of $\mathbf{P}^{\top}\mathbf{P}$ are equal to 1. And the mutual coherence of $\mathbf{P} \in \Re^{m \times n}$, which is the same as the definition in \cite{NEURIPS2018}, is defined as:
	\begin{equation}
		\mu_1\triangleq\mu_1(\mathbf{P})=\max _{\substack{i \neq j \\ 1 \leq i, j \leq n}}\left|\left(\mathbf{P}_i\right)^{\top} \mathbf{P}_j\right|,
	\end{equation}
where $\mathbf{P}_i$ refers to the $i^{th}$ column of matrix $\mathbf{P}$.

We further define a new mutual coherence, in which we do not consider the position of those zero elements known by prior information :
\begin{equation}
    \mu_2\triangleq\mu_2(\mathbf{P})=\max _{\substack{i \neq j \\ i \notin 
    \Lambda}}\left|\left(\mathbf{P}_i\right)^{\top} \mathbf{P}_j\right|.
\end{equation}
\end{definition}
Since in $\mu_2$ some columns are not included in the calculation of the maximum value, we have $\mu_2 \leq \mu_1$.

\begin{thm}[A-PIAAE Convergence Guarantee]
	If definition 1 holds and $s$ is sufficiently small, then there exists sequence of parameters $\left\{\mathbf{P}, \theta^l\right\}_{l=0}^{\infty}$ of A-PIAAE that satisfy
	\begin{equation}
		C_P=\max _{1 \leq i, j \leq N}\left|\mathbf{P}_{i, j}\right|,
	\end{equation}
	\begin{equation}\label{theta}
		\theta^l=\sup _{\left(\mathbf{h}^*, \mathbf{n}\right) \in \mathcal{H}(B, s, \sigma)}\left\{{\mu}_2\left\|\mathbf{h}^l-\mathbf{h}^*\right\|_1\right\}+C_P \sigma,
	\end{equation}
	such that, for all $\left( {{\mathbf{h}^*},\mathbf{n} } \right) \in {\cal H}(B,s,\sigma )$, we have the error bound:
	\begin{equation}\label{conclusion}
		\left\|\mathbf{h}^l-\mathbf{h}^*\right\|_2 \leq s B \exp (-c l)+C \sigma, \quad \forall l=1,2, \cdots,
	\end{equation}
	where $c=-log(\mu_1s-\mu_1+\mu_2s)$ and $C=\frac{2 s C_P}{1-\mu_1s-\mu_2s+\mu_1}$.
\end{thm}

Theorem 2 demonstrates the performance gain of the prior information in A-PIAAE. In order to satisfy the need to make the recovery error in \eqref{conclusion} decrease as the number of layers $l$ increases, we need to satisfy that $c > 0$. Then, we have $\mu_1s-\mu_1+\mu_2s<1$. We know from the prior information that the elements at some position in the vector must be zero, as we define the parameter $\mu_2$. Recall the parameters defined in \cite{NEURIPS2018}, $\theta^l=\sup _{\left(\mathbf{h}^*, \mathbf{n}\right) \in \mathcal{H}(B, s, \sigma)}\left\{{\mu}_1\left\|\mathbf{h}^l-\mathbf{h}^*\right\|_1\right\}+C_P \sigma$. Given that $\mu_2 \leq \mu_1$, the selection of the parameter $\theta^l$ in \eqref{theta} becomes more flexible, allowing a greater range of values of $\theta^l$ that can facilitate the convergence of the algorithm. Furthermore, the error bound of LISTA in \cite{NEURIPS2018} has the same form as in \eqref{conclusion}, with $c=-log(2\mu_1s-\mu_1)$ and $C=\frac{2 s C_P}{1-2\mu_1s+\mu_1}$. Since $c,C$ are smaller due to the condition $\mu_2 \leq \mu_1$, the error bound is reduced compared to LISTA, indicating that the incorporation of prior information enhances the performance of the algorithm. Furthermore, while the prior information is specifically tailored for MDs, it contributes to an overall enhancement in system performance. The non-orthogonality between the two device types enables the prior information utilized for MDs to also improve the AUD performance of ADs, thereby benefiting both device categories.

We further analyze the computational complexity of ISTA, LISTA and our proposed A-PIAAE. 
The computational complexity of one iteration of ISTA is $\mathcal{O}(S^2M)$. This cost arises primarily from the matrix multiplication $\mathbf{P}^\top\mathbf{P}$. In contrast, LISTA replaces the matrix multiplication with learnable parameter matrices $\mathbf{W}_1$ and $\mathbf{W}_2$. Consequently, each layer of LISTA incurs a lower complexity of $\mathcal{O}(S^2)$ due to the fixed learned weights. The computational process of A-PIAAE in the inference stage is identical to ISTA, therefore, its computational complexity of one iteration is also $\mathcal{O}(S^2M)$. 
However, A-PIAAE achieves comparable or better recovery performance within just $L=15$ layers, as opposed to over 1,000 ISTA iterations typically required for convergence. This translates into a significant reduction in inference time, making the proposed method well-suited for real-time implementation in grant-free mMTC systems. In addition, the network parameters in A-PIAAE are fixed post-training, eliminating any runtime overhead beyond matrix multiplications and thresholding operations.

\begin{figure}
	\centerline{\includegraphics[scale=0.55]{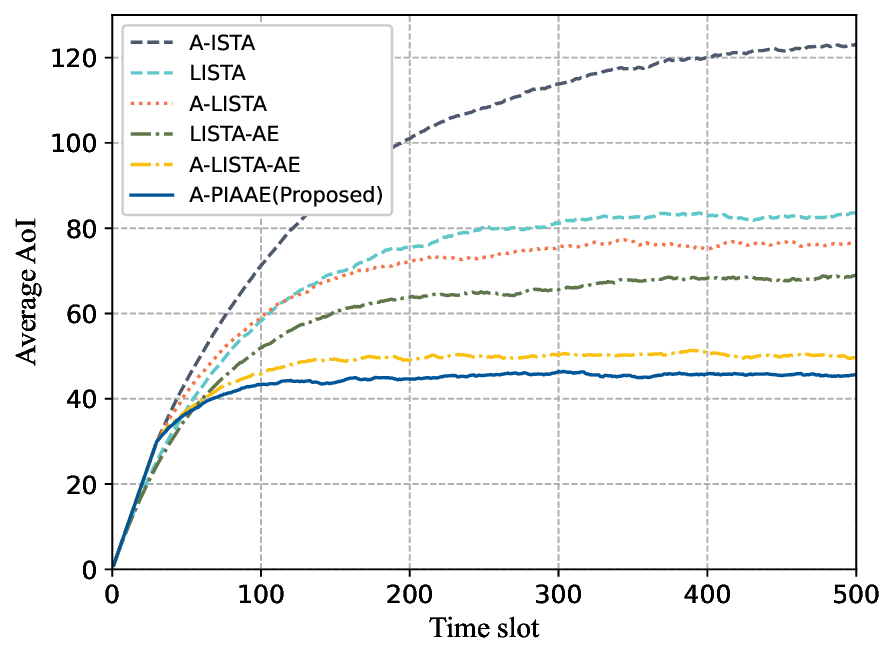}}
	\caption{Evolution of system average AoI over time. The number of ADs is $N=64$, the number of MDs is $K=128$, the length of pilot is $M=39$ and $\mathrm{SNR}=20\,  \mathrm{dB}$.}
	\label{aaoi_time}
\end{figure}

From the deployment perspective, the method scales efficiently with massive device numbers. For instance, when the total number of devices reaches tens of thousands, these devices can be partitioned across orthogonal resource blocks, and A-PIAAE can be independently applied to each block. This parallelizable structure facilitates distributed implementation across edge computing units or baseband processing modules at the BS.



\section{Experiment}
In this section, we construct various experiments to demonstrate the superiority of the optimized ARA scheme and the proposed A-PIAAE, which are designed to minimize the average AoI of MDs and the recovery error of ADs in the mMTC scenario. Firstly, we show the evolution of average AoI of MDs over time of different algorithms and demonstrate the improvement of A-PIAAE in minimizing average AoI of MDs. Then, we use A-PIAAE for channel estimation and AUD, and show the performance gain under different pilot lengths and SNRs. As we construct our network based on ISTA, the compared methods are elaborated as follows:

\begin{itemize}
	\item Age-based Random Access ISTA (A-ISTA): Using ARA scheme and doing AUD with ISTA which is the traditional mathematical solution.
	\item A-LISTA: Using ARA scheme and doing AUD with LISTA.
	\item A-LISTA with autoencoder (A-LISTA-AE): Using ARA scheme. An autoencoder is employed for AUD and pilot design, where the decoder uses LISTA.
	\item Random Access LISTA (LISTA): Using a random access scheme in which all devices can access the channel with probability $p$. Using LISTA for AUD.
	\item Random Access LISTA-AE (LISTA-AE): The access scheme is the same as that of LISTA. Using an autoencoder for AUD and pilot design. The decoder uses LISTA.
\end{itemize}

\begin{figure}
	\centerline{\includegraphics[scale=0.55]{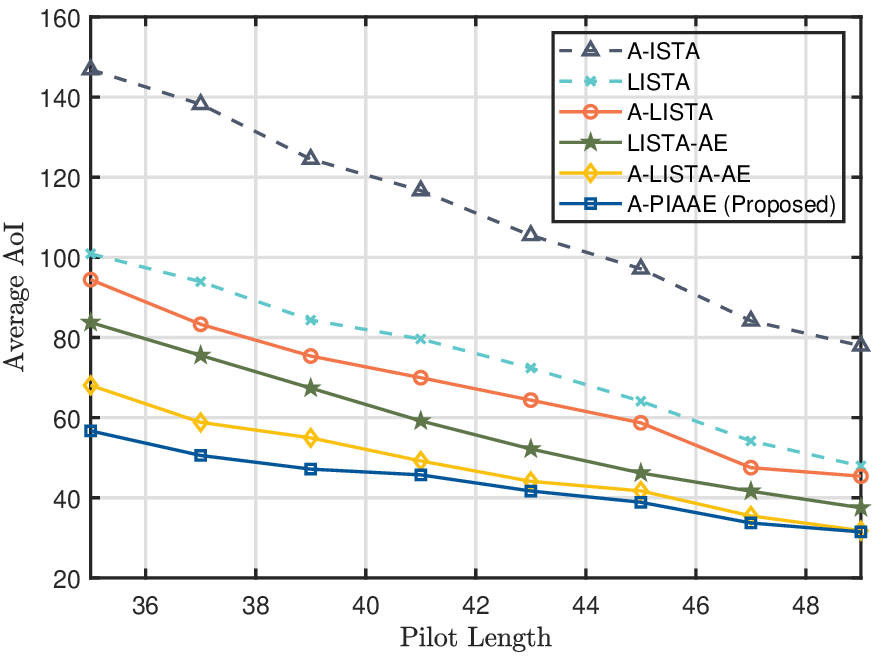}}
	\caption{Stationary system average AoI of different methods under different lengths of pilots. The number of ADs is $N=64$, the number of MDs is $K=128$ and $\mathrm{SNR}=20\,  \mathrm{dB}$.}
	\label{aaoi_pl}
\end{figure}
\textbf{Experiment Setting.} In the experiment, we set the number of ADs $N$ to 64 and the number of MDs $K$ to 128. The non-orthogonal pilot matrix $\mathbf{P}$ is generated from the Gaussian distribution and then normalized column by column. We assume the maximum AoI $a_{\max}$ of every MD is 100 for generating the data sets. Each entry in $\mathbf{g}_t$ is independently non-zero with probability $p$.
	Similarly, each entry in $\mathbf{x}_t$ has a probability of 0.05 to be non-zero. And we assume that the wireless channel is Rayleigh fading channel, so the non-zero values of $\mathbf{g}_t$ and $\mathbf{x}_t$ follow the Gaussian distribution. We apply the age-dependent random access scheme in OFDM systems, with a total of 1.4 MHz bandwidth and 72 subcarriers. For network training, we adopt online training with batch size 64 and adopt a stage-wise training method according to \cite{7934066}. The initial learning rate is 0.001 and is reduced by a certain percentage when the loss does not decrease for 500 iterations. The learning rate decay is applied a total of three times, with the learning rate successively scaled by factors of 0.5, 0.1, and 0.01. We set the layers of all NN to be 15, in which the networks have converged to a satisfactory level of performance, and set the iteration number of ISTA to be 1000 for comparison.  The maximum permissible error $\tau$ in \eqref{criteria} is 0.1 to simulate the evolution of AoI over time. To evaluate the AUD performance of ADs for different methods, we use successful detection rate as the evaluation index, which is defined as 
\begin{equation}
		r=\frac{\operatorname{card}(\mathcal{N}_t \bigcap \hat{\mathcal{N}}_t)}{N_t},
\end{equation}
where $\operatorname{card}()$ denotes the number of elements in the set.

\begin{table*}[t]
	\centering
	\caption{The successful detection rate of ADs of different algorithms under different pilot length [\%]. The numerical values in parentheses denote the optimal threshold $\delta$ and access probability $p$ of ARA schemes.}
	\label{pl_nmse_c}
	\resizebox{0.8\textwidth}{!}{
		\begin{tabular}{cccccccccc}
				\hline
				\toprule[1pt]
				\multicolumn{1}{c|}{Pilot Length} & \multicolumn{1}{c|}{35} & \multicolumn{1}{c|}{37} & \multicolumn{1}{c|}{39} & \multicolumn{1}{c|}{41} & \multicolumn{1}{c|}{43} & \multicolumn{1}{c|}{45} & \multicolumn{1}{c|}{47} & \multicolumn{1}{c}{49}\\ 
				\multicolumn{1}{c|}{($\delta$, $p$)} 
				& \multicolumn{1}{c|}{(43, 0.05)} & \multicolumn{1}{c|}{(43, 0.05)} 
				& \multicolumn{1}{c|}{(29, 0.05)} & \multicolumn{1}{c|}{(18, 0.05)} 
				& \multicolumn{1}{c|}{(18, 0.05)} & \multicolumn{1}{c|}{(11, 0.05)} 
				& \multicolumn{1}{c|}{(11, 0.06)} & \multicolumn{1}{c}{(11, 0.06)} \\
				\midrule
				\multicolumn{1}{c|}{A-ISTA}      & \multicolumn{1}{c|}{0.166}   & \multicolumn{1}{c|}{0.171}   & \multicolumn{1}{c|}{0.192}   & \multicolumn{1}{c|}{0.197}   & \multicolumn{1}{c|}{0.218}   & \multicolumn{1}{c|}{0.234}   &
				\multicolumn{1}{c|}{0.227}   &
				\multicolumn{1}{c}{0.242}   &\\
				\multicolumn{1}{c|}{LISTA}      & \multicolumn{1}{c|}{0.210}   & \multicolumn{1}{c|}{0.223}   & \multicolumn{1}{c|}{0.256}   & \multicolumn{1}{c|}{0.265}   & \multicolumn{1}{c|}{0.286}   & \multicolumn{1}{c|}{0.329}   &
				\multicolumn{1}{c|}{0.310}   &
				\multicolumn{1}{c}{0.349}   &\\
				\multicolumn{1}{c|}{LISTA-AE}         & \multicolumn{1}{c|}{0.260}   & \multicolumn{1}{c|}{0.283}   & \multicolumn{1}{c|}{0.322}   & \multicolumn{1}{c|}{0.365}   & \multicolumn{1}{c|}{0.412}   & \multicolumn{1}{c|}{0.461} &
				\multicolumn{1}{c|}{0.422} &
				\multicolumn{1}{c}{0.470} &
				\\
				\multicolumn{1}{c|}{A-LISTA}        & \multicolumn{1}{c|}{0.285}   & \multicolumn{1}{c|}{0.332}   & \multicolumn{1}{c|}{0.322}   & \multicolumn{1}{c|}{0.331}   & \multicolumn{1}{c|}{0.363}   & \multicolumn{1}{c|}{0.365}   &
				\multicolumn{1}{c|}{0.369}   &
				\multicolumn{1}{c}{0.393}   &\\
				\multicolumn{1}{c|}{A-LISTA-AE}      & \multicolumn{1}{c|}{0.386}   & \multicolumn{1}{c|}{0.471}   & \multicolumn{1}{c|}{0.439}   & \multicolumn{1}{c|}{0.457}   & \multicolumn{1}{c|}{0.532}   & \multicolumn{1}{c|}{0.521}  &
				\multicolumn{1}{c|}{0.522}  &
				\multicolumn{1}{c}{0.577}  &\\
				\multicolumn{1}{c|}{\textbf{A-PIAAE (Proposed)}}    & \multicolumn{1}{c|}{\textbf{0.479}}   & \multicolumn{1}{c|}{\textbf{0.568}}   & \multicolumn{1}{c|}{\textbf{0.536}}   & \multicolumn{1}{c|}{\textbf{0.501}}   & \multicolumn{1}{c|}{\textbf{0.560}}   & \multicolumn{1}{c|}{\textbf{0.569}}   &
				\multicolumn{1}{c|}{\textbf{0.550}}   &
				\multicolumn{1}{c}{\textbf{0.588}}   &
				\\
				\bottomrule[1pt]
		\end{tabular}
		}
\end{table*}

\begin{figure}[t]
	\centering
	\begin{minipage}[b]{1\linewidth}
		\centering
		\subfigure[AUD performance of ADs]{
			\includegraphics[width=0.9\linewidth]{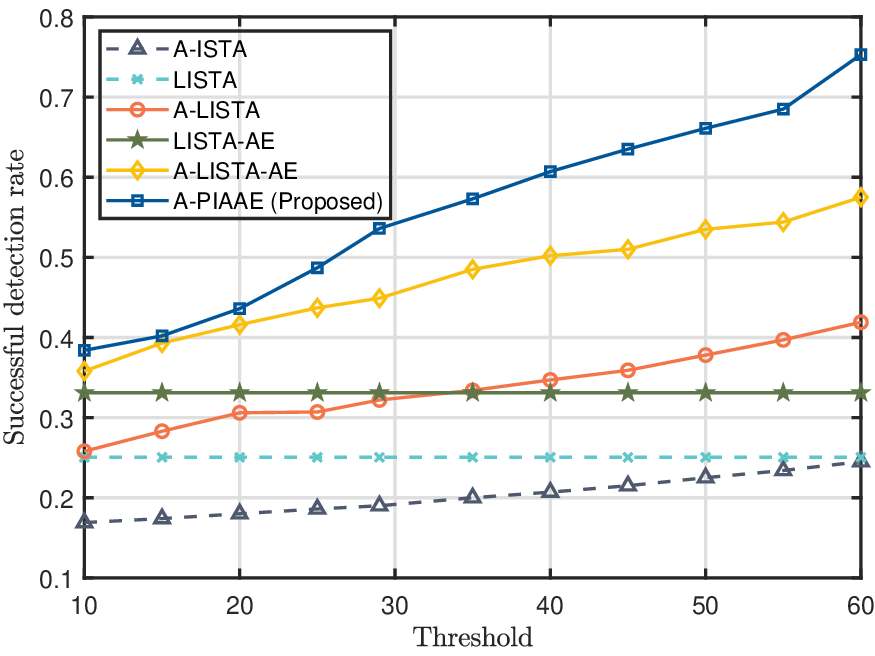}
		}
		\subfigure[Average AoI of MDs]{
			\includegraphics[width=0.9\linewidth]{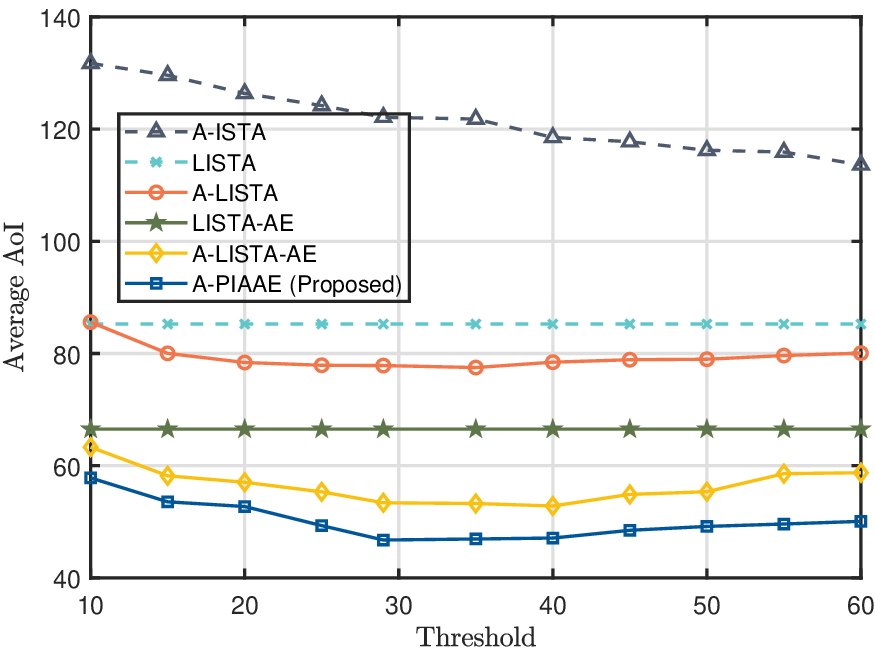}
		}
	\end{minipage}
	\caption{Successful detection rate of ADs and stationary average AoI of MDs of different methods under different threshold $\delta$. The number of ADs is $N=64$, the number of MDs is $K=128$, the length of pilot is $M=39$ and $\mathrm{SNR}=20$.}
	\label{threshold}
\end{figure}

First, we compare the evolution of the average AoI of MDs over time for different algorithms. As shown in Fig. \ref{aaoi_time}, the average AoI for all four ARA algorithms initially increases linearly, which means AoI of all MDs are below the threshold and all MDs are inactive. As time progresses, all six algorithms reach a stationary state, with the proposed algorithm A-PIAAE achieving the lowest average AoI. The effectiveness of the proposed ARA scheme, prior information aided NN and the design of autoencoder can be validated from the comparison of the algorithms. In the next experiments, we use the average AoI of the stationary state to compare all algorithms.

In Table \ref{pl_nmse_c}, we show the optimal AoI threshold $\delta$, active probability $p$, which are solved by the method introduced in Section III.A, and the successful detection rate of ADs under different lengths of pilots. As AUD performance can be improved with more pilot resources, the limitation of the access threshold $\delta$ decreases, which means more devices with smaller AoI can access the channel as the pilot length increases. It is important to note that the performance of the AUD does not improve linearly with increasing pilot length. This non-linear relationship arises because changes in pilot length affect access parameters, which in turn alters signal sparsity. Nevertheless, A-PIAAE always has the highest successful detection rate.

As shown in Fig. \ref{aaoi_pl}, we compare the average AoI of MDs under different pilot lengths. The average AoI of all six algorithms decreases as the pilot length increases, which means that more resources can guarantee recovery accuracy and reduce the average AoI of MDs. The proposed A-PIAAE has the lowest average AoI, and the gain is more pronounced with shorter pilot length. This phenomenon occurs because with shorter pilot length, the threshold of age-based access scheme is larger and more devices are denied to access the channel, giving A-PIAAE more prior information and leads to greater performance gain. 

\begin{figure}[t]
	\centering
	\begin{minipage}[b]{1\linewidth}
		\centering
		\subfigure[AUD performance of ADs.]{
			\includegraphics[width=0.9\linewidth]{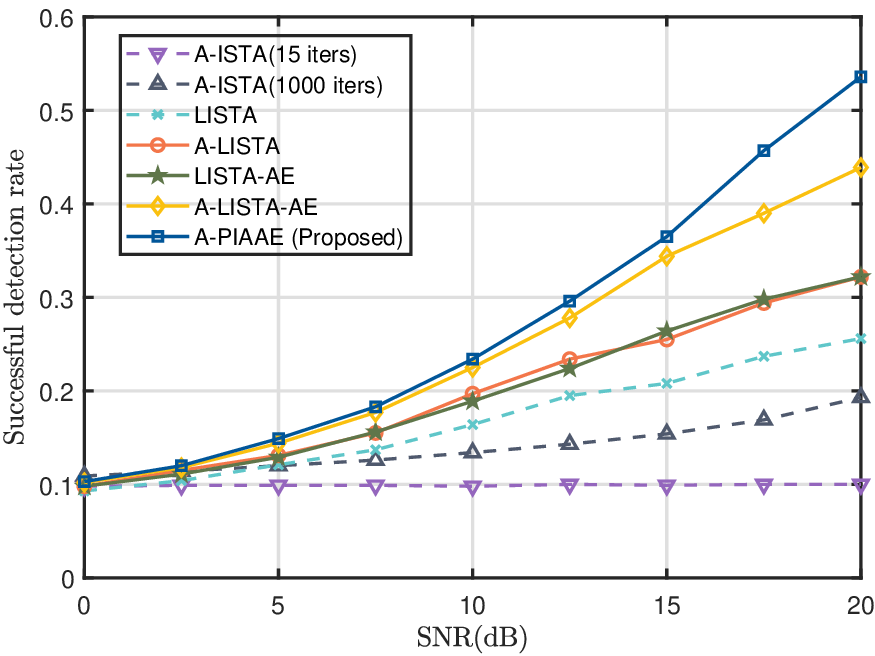}
		}
		\subfigure[Average AoI of MDs.]{
			\includegraphics[width=0.9\linewidth]{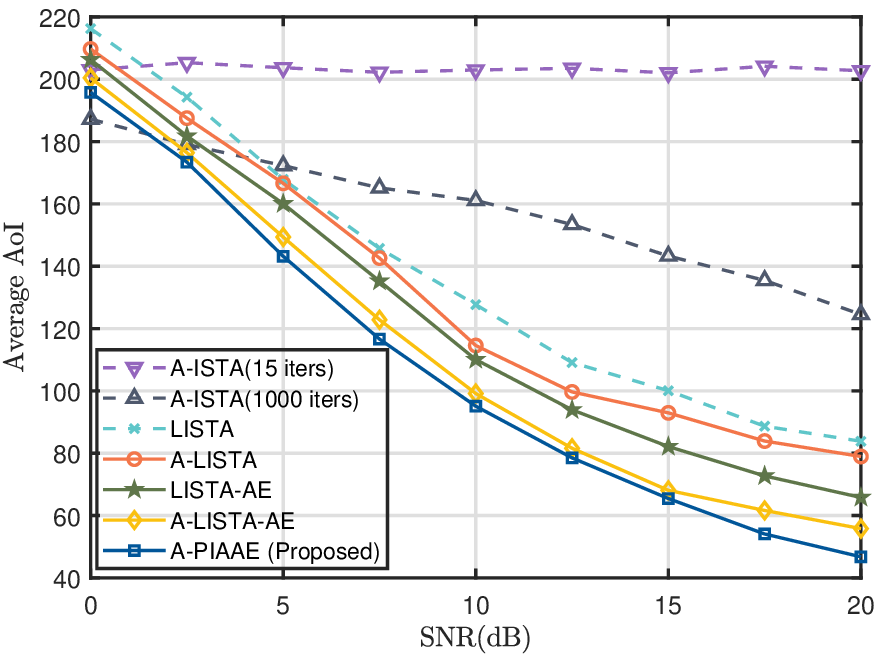}
		}
	\end{minipage}
	\caption{Successful detection rate of ADs and stationary average AoI of MDs of different methods under different SNRs. The number of ADs is $N=64$, the number of MDs is $K=128$, the length of pilot is $M=39$.}
	\label{snr}
\end{figure}

\begin{figure}[t]
	\centering
	\begin{minipage}[b]{1\linewidth}
		\centering
		\subfigure[Average AoI of MDs with the ratio of MDs:ADs is 2:1.]{
			\includegraphics[width=0.9\linewidth]{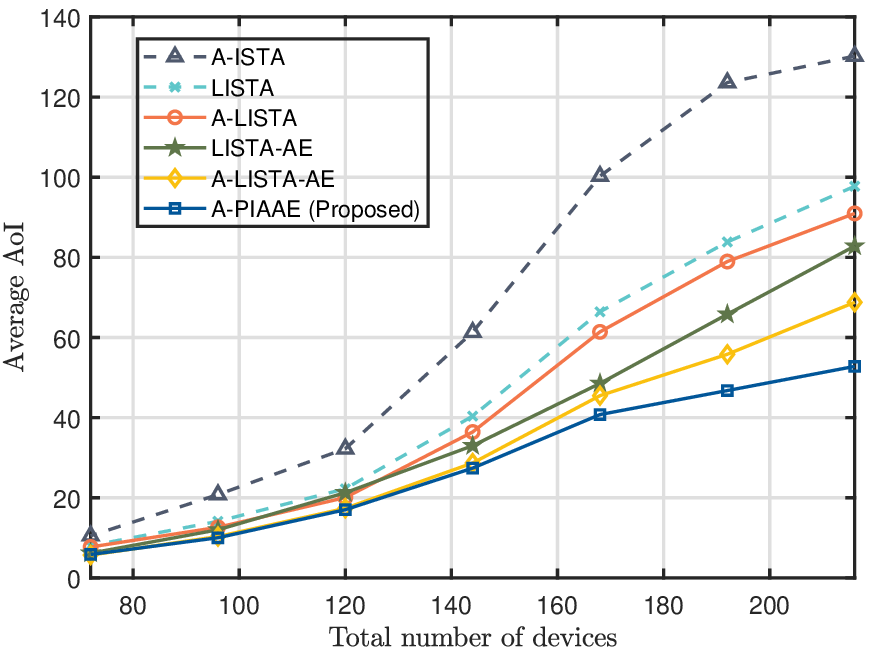}
		}
		\subfigure[Average AoI of MDs with the ratio of MDs:ADs is 1:1.]{
			\includegraphics[width=0.9\linewidth]{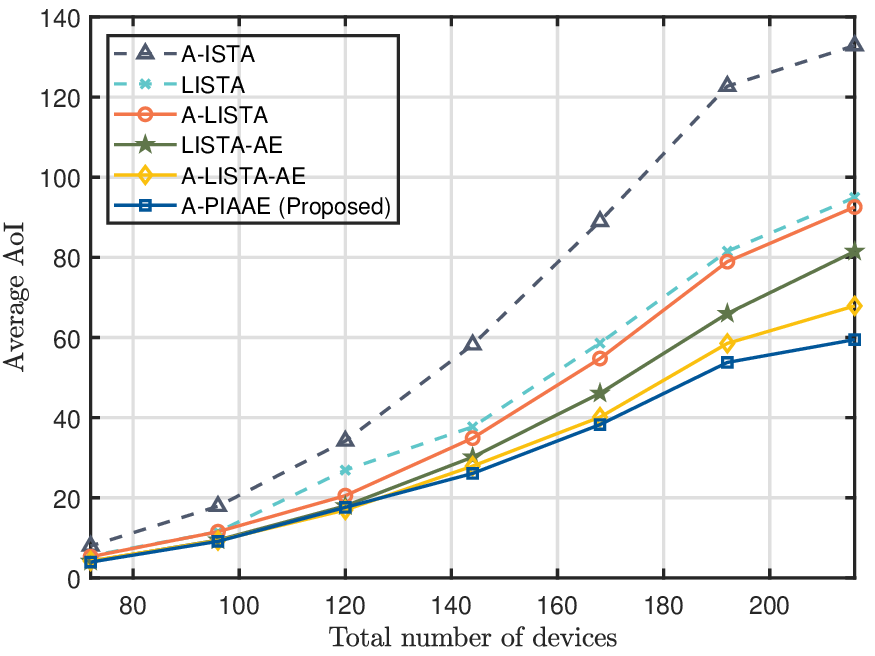}
		}
	\end{minipage}
	\caption{The stationary average AoI of MDs of different methods under different total number of devices. The length of pilot is $M=39$ and $\mathrm{SNR}=20$.}
	\label{totalnum}
\end{figure}

Fig. \ref{threshold} shows the trend of the successful detection rate of ADs and the average AoI of MDs with the change of threshold $\delta$. The AUD performance and the average AoI of two RA schemes are not affected by the change of the threshold. And the successful detection rate of ADs of the other ARA schemes (except A-ISTA, which is because the experiment with 1000 cycles of ISTA does not achieve convergence performance) increases as the access threshold becomes larger, while there exists a minimum value for the average AoI of MDs, which coincides with our theoretical analyses in section \uppercase\expandafter{\romannumeral3}. We can see that the optimal threshold obtained from theoretical analyses is slightly smaller than the optimal threshold in the actual case of using DL methods. The reason for the deviation is due to the difference in the successful detection rate obtained using DL-based methods and traditional CS theory. Notably, our A-PIAAE algorithm is optimal in terms of both AUD performance for ADs and minimizing average AoI for MDs.

In the experiments on the AUD performance of ADs and average AoI of MDs under varying SNRs, we further add A-ISTA with 15 iterations, which is the same as the network layers. As shown in Fig. \ref{snr}, when SNR is higher than 2.5, the performance of the proposed A-PIAAE is optimal; when SNR is lower than 2.5, the performance of the proposed algorithm is lower than that of ISTA with 1000 iterations but still better than that of ISTA with the same 15 iterations. However, better performance can still be achieved by increasing the number of network layers and using a larger training set.

Fig. \ref{totalnum} illustrates the average AoI performance under different total device numbers when MDs:ADs is 2:1 and 1:1. As can be seen from the figure, no matter how the ratio of the two types of devices changes, the overall trend of the average AoI changing with the total number of devices is the same, and the proposed A-PIAAE has the best performance. As the total number of devices increases, the optimal access parameters impose greater restrictions on the access of MD devices, and the access threshold $\delta$ becomes larger, which means the proposed A-PIAAE can use more prior information, so the performance improvement is also greater.



\section{Conclusion}
In this paper, we have investigated the random access problem for GR-RA with heterogeneous traffic. We have analyzed the ARA scheme, derived the expression for the average AoI and optimized the access parameters. Then, we have designed an autoencoder A-PIAAE to jointly optimize the pilots and learn the detection model, which can reduce interference between pilots and improve access performance for both types of devices. Utilizing the AoI of MDs as prior information, we have designed the decoder called LISTA-AGE to further enhance the recovery performance. Furthermore, we have theoretically analyzed the convergence of the proposed A-PIAAE. According to our experimental results, the proposed method outperformed the traditional method and existing approaches in both accuracy and information freshness.
\appendix
\section*{Proof of The Theorem 2}
The proof is organized in three parts: First, we proof that the value of $\mathbf{h}^l$ obtained through A-PIAAE must necessarily be zero at the positions corresponding to the zero elements in the ground truth $\mathbf{h}^*$. Then, we derive the recovery error for one data in $\mathcal{H}(B, s, \sigma)$. Last, we get the error bound for the whole data set.
\newtheorem{lemma}{Lemma}
	\begin{lemma}
		We take $\left(\mathbf{h}^*, \mathbf{n}\right) \in \mathcal{H}(B, s, \sigma)$ and let $S=\mathrm{support}\left(\mathbf{h}^*\right)$. If \eqref{theta} holds, then we have $\mathbf{h}_i^l=0, \forall i \notin S, \forall l$
\end{lemma}

\begin{proof}
	When $i \in \Lambda$, due to the access scheme in which the elements of corresponding age below the threshold are definitely zero, we have $\mathbf{h}_i^l=0$.
	
	When $i \notin \Lambda$ and $i \notin S$, we prove this by induction. When $l=0$, it is satisfied since $\mathbf{h}^0=0$. Fixing $l$, and assuming $\mathbf{h}_i^l=0, \forall i \notin S, i \notin \Lambda$, we have 
	\begin{equation}
		\begin{aligned}
			\mathbf{h}_i^{l+1} & =\eta_{\theta^l}\left(\mathbf{h}_i^l-\sum_{j \in S}\left(\mathbf{P}_i\right)^\top\left(\mathbf{P} \mathbf{h}^l-\mathbf{y}\right)\right) \\
			& =\eta_{\theta^l}\left(-\sum_{j \in S}\left(\mathbf{P}_i\right)^\top \mathbf{P}_j\left(\mathbf{h}_j^l-\mathbf{h}_j^*\right)+\left(\mathbf{P}_i\right)^\top \mathbf{n}\right).
		\end{aligned}
	\end{equation}
	Since $\theta^l=\sup _{\left(\mathbf{h}^*, \mathbf{n}\right) \in \mathcal{H}(B, s, \sigma)}\left\{{\mu}_2\left\|\mathbf{h}^l-\mathbf{h}^*\right\|_1\right\}+C_P \sigma$
	\begin{equation}
		\begin{aligned}
			\theta^l  & \geq {\mu}_2\left\|\mathbf{h}^l-\mathbf{h}^*\right\|_1+C_P\|\mathbf{n}\|_1 \\
			& \geq\left|-\sum_{j \in S}\left(\mathbf{P}_i\right)^\top \mathbf{P}_j\left(\mathbf{h}_j^l-\mathbf{h}_j^*\right)+\left(\mathbf{P}_i\right)^\top \mathbf{n}\right|, \forall i \notin S,
		\end{aligned}
	\end{equation}
	which, implies $\mathbf{h}_i^{l+1}=0, \forall i \notin S, i \notin \Lambda$ by the definition of the soft-thresholding function $ \eta_{\theta^l}$. By induction, we have $\mathbf{h}_i^l=0, \forall i \notin S,i \notin \Lambda, \forall l$.
	
	In summary, we have
	\begin{equation}\label{step1}
		\mathbf{h}_i^l=0, \forall i \notin S, \forall l.
	\end{equation}
\end{proof}

\begin{lemma}
		For one $\left(\mathbf{h}^*, \mathbf{n}\right)\in \mathcal{H}(B, s, \sigma)$, we have 
		\begin{equation}
			\left\|\mathbf{h}^{l+1}-\mathbf{h}^*\right\|_1 \leq {\mu}_1(|S|-1)\left\|\mathbf{h}^l-\mathbf{h}^*\right\|_1+\theta^l|S|+|S| C_P \sigma,
		\end{equation}
		where $|S|$ denotes the number of elements in set $S$
\end{lemma}

\begin{proof}
	According to Lemma 1, we only need to consider the recovery error of the elements on $S$. For all $i \in S$,
	\begin{equation}
		\begin{aligned}
			\mathbf{h}_i^{l+1} & =\eta_{\theta^l}\left(\mathbf{h}_i^l-\left(\mathbf{P}_i\right)^\top \mathbf{P}_S\left(\mathbf{h}_S^l-\mathbf{h}_S^*\right)+\left(\mathbf{P}_i\right)^\top \mathbf{n}\right) \\
			& \in \mathbf{h}_i^l-\left(\mathbf{P}_i\right)^T \mathbf{P}_S\left(\mathbf{h}_S^l-\mathbf{h}_S^*\right)+\left(\mathbf{P}_i\right)^\top \mathbf{n}-\theta^l \partial \ell_1\left(\mathbf{h}_i^{l+1}\right),
		\end{aligned}
	\end{equation}
	where $\partial \ell_1\left(\mathbf{h}\right)$ is the sub-gradient of $\|\mathbf{h}\|_1$ and defined component-wisely as
	\begin{equation}\label{sub-gra}
		\partial \ell_1(\mathbf{h})_i= \begin{cases}\left\{\operatorname{sign}\left(\mathbf{h}_i\right)\right\} & \text { if } \mathbf{h}_i \neq 0 \\ {[-1,1]} & \text { if } \mathbf{h}_i=0\end{cases}.
	\end{equation}
	Since $\mathbf{P}_i^\top\mathbf{P}_i=1$, we have
	\begin{equation}
		\begin{aligned}
			&\mathbf{h}_i^l-\left(\mathbf{P}_i\right)^\top \mathbf{P}_S\left(\mathbf{h}_S^l-\mathbf{h}_S^*\right)\\  =&\mathbf{h}_i^l-\sum_{j \in S, j \neq i}\left(\mathbf{P}_i\right)^\top \mathbf{P}_j\left(\mathbf{h}_j^l-\mathbf{h}_j^*\right)-\left(\mathbf{h}_i^l-\mathbf{h}_i^*\right) \\
			=&\mathbf{h}_i^*-\sum_{j \in S, j \neq i}\left(\mathbf{P}_i\right)^\top \mathbf{P}_j\left(\mathbf{h}_j^l-\mathbf{h}_j^*\right) .
		\end{aligned}
	\end{equation}
	Then for $\forall i \in S$,
	\begin{equation}
		\begin{aligned}
			\mathbf{h}_i^{l+1}-\mathbf{h}_i^* \in&-\sum_{j \in S, j \neq i}\left(\mathbf{P}_i\right)^\top \mathbf{P}_j\left(\mathbf{h}_j^l-\mathbf{h}_j^*\right)\\&+\left(\mathbf{P}_i\right)^\top \mathbf{n}-\theta^l \partial \ell_1\left(\mathbf{h}_i^{l+1}\right).
		\end{aligned}
	\end{equation}
	By the definition \eqref{sub-gra}, every element in $\partial \ell_1(\mathbf{h}), \forall \mathbf{h} \in \Re$ has a magnitude less than or equal to 1. Thus, for $\forall i \in S$,
	\begin{equation}
		\begin{aligned}
			&\left|\mathbf{h}_i^{l+1}-\mathbf{h}_i^*\right| \\ \leq& \sum_{j \in S, j \neq i}\left|\left(\mathbf{P}_i\right)^\top \mathbf{P}_j\right|\left|\mathbf{h}_j^l-\mathbf{h}_j^*\right|+\theta^l+\left|\left(\mathbf{P}_i\right)^\top \mathbf{n}\right| \\
			\leq& \mu_1 \sum_{j \in S, j \neq i}\left|\mathbf{h}_j^l-\mathbf{h}_j^*\right|+\theta^l+C_P\|\mathbf{n}\|_1.
		\end{aligned}
	\end{equation}
	The conclusion \eqref{step1} of step 1 implies $\left\|\mathbf{h}^l-\mathbf{h}^*\right\|_1=\left\|\mathbf{h}_S^l-\mathbf{h}_S^*\right\|_1$ for all $l$. Then
	\begin{equation}
		\begin{aligned}
			&\left\|\mathbf{h}^{l+1}-\mathbf{h}^*\right\|_1= \sum_{i \in S}\left|\mathbf{h}_i^{l+1}-\mathbf{h}_i^*\right| \\
			& \leq \sum_{i \in S}\left({\mu}_1 \sum_{j \in S, j \neq i}\left|\mathbf{h}_j^l-\mathbf{h}_j^*\right|+\theta^l+C_P \sigma\right) \\
			& ={\mu}_1(|S|-1) \sum_{i \in S}\left|\mathbf{h}_i^l-\mathbf{h}_i^*\right|+\theta^l|S|+|S| C_P \sigma \\
			& \leq {\mu}_1(|S|-1)\left\|\mathbf{h}^l-\mathbf{h}^*\right\|_1+\theta^l|S|+|S| C_P \sigma.
		\end{aligned}
	\end{equation}
\end{proof}

Based on Lemma 2, we then derive the error bound for the whole data set. We take supremum over $\left( {{\mathbf{h}^*},\mathbf{n} } \right) \in {\cal H}(B,s,\sigma )$, by $|S| \leq s$,
\begin{equation}
	\begin{aligned}
		\sup _{\mathbf{h}^*, \mathbf{n}}\left\{\left\|\mathbf{h}^{l+1}-\mathbf{h}^*\right\|_1\right\} \leq & {\mu}_1(s-1) \sup _{\mathbf{h}^*, \mathbf{n}}\left\{\left\|\mathbf{h}^l-\mathbf{h}^*\right\|_1\right\}\\
		&+s \theta^l+s C_P \sigma .
	\end{aligned}
\end{equation}
By $\theta^l=\sup _{\mathbf{h}^*, \mathbf{n}}\left\{{\mu}_2\left\|\mathbf{h}^l-\mathbf{h}^*\right\|_1\right\}+C_P \sigma$, we have
\begin{equation}
	\begin{split}
		\sup _{\mathbf{h}^*, \mathbf{n}}\left\{\left\|\mathbf{h}^{l+1}-\mathbf{h}^*\right\|_1\right\} \leq &(\mu_1s-\mu_1+\mu_2s) \sup _{\mathbf{h}^*, \mathbf{n}}\left\{\left\|\mathbf{h}^l-\mathbf{h}^*\right\|_1\right\}\\
		&+2 s C_P \sigma.
	\end{split}
\end{equation}
By induction, with $c=-log(\mu_1s-\mu_1+\mu_2s)$, $C=\frac{2 s C_P}{1-\mu_1s-\mu_2s+\mu_1}$, we obtain 
\begin{equation}
	\begin{aligned}
		&\sup _{\mathbf{h}^*, \mathbf{n}}\left\{\left\|\mathbf{h}^{l+1}-\mathbf{h}^*\right\|_1\right\} \\ \leq&(\mu_1s-\mu_1+\mu_2s)^{l+1} \sup _{\mathbf{h}^*, \mathbf{n}}\left\{\left\|\mathbf{h}^0-\mathbf{h}^*\right\|_1\right\}\\
		&+2 s C_P \sigma\left(\sum_{\tau=0}^{l+1}(\mu_1s-\mu_1+\mu_2s)^\tau\right) \\
		& \leq(\mu_1s-\mu_1+\mu_2s)^l s B+C \sigma\\
		&=s B \exp (-c l)+C \sigma .
	\end{aligned}
\end{equation}
Since $\|\mathbf{h}\|_2 \leq\|\mathbf{h}\|_1$ for any $\mathbf{h} \in \Re$, we can get the upper bound for $\ell_2$ norm:
\begin{equation}
	\begin{aligned}
		\sup _{\mathbf{h}^*, \mathbf{n}}\left\{\left\|\mathbf{h}^{l+1}-\mathbf{h}^*\right\|_2\right\} &\leq \sup _{\mathbf{h}^*, \mathbf{n}}\left\{\left\|\mathbf{h}^{l+1}-\mathbf{h}^*\right\|_1\right\} \\&\leq s B \exp (-c l)+C \sigma .
	\end{aligned}
\end{equation}
As long as $s \le \frac{{1 + {\mu _1}}}{{{\mu _1} + {\mu _2}}}$, $c=-log(\mu_1s-\mu_1+\mu_2s) \ge 0$, then the error bound \eqref{conclusion} holds uniformly for all $\left(\mathbf{h}^*, \mathbf{n}\right) \in \mathcal{H}(B, s, \sigma)$.

\bibliographystyle{IEEEtran}
\bibliography{IEEEabrv,ref_2}

\begin{thebibliography}{10}
\providecommand{\url}[1]{#1}
\csname url@samestyle\endcsname
\providecommand{\newblock}{\relax}
\providecommand{\bibinfo}[2]{#2}
\providecommand{\BIBentrySTDinterwordspacing}{\spaceskip=0pt\relax}
\providecommand{\BIBentryALTinterwordstretchfactor}{4}
\providecommand{\BIBentryALTinterwordspacing}{\spaceskip=\fontdimen2\font plus
\BIBentryALTinterwordstretchfactor\fontdimen3\font minus \fontdimen4\font\relax}
\providecommand{\BIBforeignlanguage}[2]{{%
\expandafter\ifx\csname l@#1\endcsname\relax
\typeout{** WARNING: IEEEtran.bst: No hyphenation pattern has been}%
\typeout{** loaded for the language `#1'. Using the pattern for}%
\typeout{** the default language instead.}%
\else
\language=\csname l@#1\endcsname
\fi
#2}}
\providecommand{\BIBdecl}{\relax}
\BIBdecl

\bibitem{10734800}
Z.~Sun, W.~Chen, Y.~Sun, X.~Qin, T.~Hou, L.~Li, and B.~Ai, ``Deep learning-based grant-free massive access with age-of-information minimization,'' in \emph{2024 IEEE 34th Int. Workshop on Mach. Learn. for Signal Process. (MLSP)}, 2024, pp. 1--6.

\bibitem{2004.14146}
N.~H. Mahmood, S.~Böcker, A.~Munari, F.~Clazzer, I.~Moerman, K.~Mikhaylov, O.~Lopez, O.-S. Park, E.~Mercier, H.~Bartz, R.~Jäntti, R.~Pragada, Y.~Ma, E.~Annanperä, C.~Wietfeld, M.~Andraud, G.~Liva, Y.~Chen, E.~Garro, F.~Burkhardt, H.~Alves, C.-F. Liu, Y.~Sadi, J.-B. Dore, E.~Kim, J.~Shin, G.-Y. Park, S.-K. Kim, C.~Yoon, K.~Anwar, and P.~Seppänen, ``White paper on critical and massive machine type communication towards 6{G},'' 2020.

\bibitem{qos-difference}
S.~Dilek, K.~Irgan, M.~Guzel, S.~Ozdemir, S.~Baydere, and C.~Charnsripinyo, ``{QoS}-aware {IoT} networks and protocols: A comprehensive survey,'' \emph{Int. J. of Communication Systems}, vol.~35, no.~10, p. e5156, 2022.

\bibitem{ABBAS2023199}
Q.~Abbas, S.~A. Hassan, H.~K. Qureshi, K.~Dev, and H.~Jung, ``A comprehensive survey on age of information in massive {IoT} networks,'' \emph{Computer Commun.}, vol. 197, pp. 199--213, 2023.

\bibitem{5984917}
S.~Kaul, M.~Gruteser, V.~Rai, and J.~Kenney, ``Minimizing age of information in vehicular networks,'' in \emph{2011 8th Annual IEEE Commun. Society Conf. on Sensor, Mesh and Ad Hoc Commun. and Netw.}, 2011, pp. 350--358.

\bibitem{3gpp_ts_38_211}
{3GPP TS 38.211 V17.3.0}, ``{{NR} Physical channels and modulation},'' Jun. 2022.

\bibitem{6125356}
H.~F. Schepker and A.~Dekorsy, ``Sparse multi-user detection for {CDMA} transmission using greedy algorithms,'' in \emph{2011 8th Int. Symposium on Wireless Communication Systems}, 2011, pp. 291--295.

\bibitem{bockelmann2013compressive}
C.~Bockelmann, H.~F. Schepker, and A.~Dekorsy, ``Compressive sensing based multi-user detection for machine-to-machine communication,'' \emph{Trans. Emerg. Telecommun. Technol.}, vol.~24, no.~4, pp. 389--400, 2013.

\bibitem{6629742}
H.~F. Schepker, C.~Bockelmann, and A.~Dekorsy, ``Exploiting sparsity in channel and data estimation for sporadic multi-user communication,'' in \emph{ISWCS 2013; The Tenth Int. Symposium on Wireless Communication Systems}, 2013, pp. 1--5.

\bibitem{doi:10.1137/080716542}
A.~Beck and M.~Teboulle, ``A fast iterative shrinkage-thresholding algorithm for linear inverse problems,'' \emph{SIAM J. on Imaging Sciences}, vol.~2, no.~1, pp. 183--202, 2009.

\bibitem{5503193}
D.~L. Donoho, A.~Maleki, and A.~Montanari, ``Message passing algorithms for compressed sensing: I. motivation and construction,'' in \emph{2010 IEEE Information Theory Workshop on Information Theory (ITW 2010, Cairo)}, 2010, pp. 1--5.

\bibitem{1315936}
D.~Wipf and B.~Rao, ``Sparse bayesian learning for basis selection,'' \emph{IEEE Trans. Signal Process.}, vol.~52, no.~8, pp. 2153--2164, 2004.

\bibitem{10819487}
J.~Fang, P.~Zhu, B.~Ai, F.-C. Zheng, and X.~You, ``Resource allocation for {eMBB/URLLC} coexistence in massive {MIMO} industrial automation,'' \emph{IEEE Internet Things J.}, pp. 1--1, 2024.

\bibitem{10654339}
M.~Alsenwi, E.~Lagunas, and S.~Chatzinotas, ``Distributed learning framework for {eMBB-URLLC} multiplexing in open radio access networks,'' \emph{IEEE Trans. Netw. and Service Manage.}, vol.~21, no.~5, pp. 5718--5732, 2024.

\bibitem{10517305}
M.~Katwe, K.~Singh, C.-P. Li, S.~Prakriya, B.~Clerckx, and G.~K. Karagiannidis, ``Enhanced user fairness and performance for {eMBB-URLLC} uplink traffic with rate- splitting based super-positioning,'' \emph{IEEE Trans. Wireless Commun.}, vol.~23, no.~9, pp. 12\,513--12\,530, 2024.

\bibitem{10238402}
J.~Wang, C.~Xing, and J.~Liu, ``Intelligent preamble allocation for coexistence of {mMTC/URLLC} devices: A hierarchical {Q-learning} based approach,'' \emph{China Commun.}, vol.~20, no.~8, pp. 44--53, 2023.

\bibitem{10190330}
Y.~Liu, B.~Clerckx, and P.~Popovski, ``Network slicing for {eMBB, URLLC, and mMTC}: An uplink rate-splitting multiple access approach,'' \emph{IEEE Trans. Wireless Commun.}, vol.~23, no.~3, pp. 2140--2152, 2024.

\bibitem{10670292}
L.~Valentini, E.~Bernardi, F.~Saggese, M.~Chiani, E.~Paolini, and P.~Popovski, ``Contention-based {mMTC/URLLC} coexistence through coded random access and massive {MIMO},'' \emph{IEEE J. Sel. Topics Signal Process.}, vol.~18, no.~7, pp. 1265--1280, 2024.

\bibitem{9903376}
Y.~Bai, W.~Chen, F.~Sun, B.~Ai, and P.~Popovski, ``Data-driven compressed sensing for massive wireless access,'' \emph{IEEE Commun. Magazine}, vol.~60, no.~11, pp. 28--34, 2022.

\bibitem{10786246}
W.~Chen, Y.~Liu, H.~Jafarkhani, Y.~C. Eldar, P.~Zhu, and K.~B. Letaief, ``Signal processing and learning for next generation multiple access in {6G},'' \emph{IEEE J. Sel. Topics Signal Process.}, pp. 1--33, 2024.

\bibitem{10.5555/3104322.3104374}
K.~Gregor and Y.~LeCun, ``Learning fast approximations of sparse coding,'' in \emph{Proceedings of the 27th Int. Conf. on Int. Conf. on Mach. Learn.}, ser. ICML'10.\hskip 1em plus 0.5em minus 0.4em\relax Madison, WI, USA: Omnipress, 2010, p. 399–406.

\bibitem{7934066}
M.~Borgerding, P.~Schniter, and S.~Rangan, ``{AMP}-inspired deep networks for sparse linear inverse problems,'' \emph{IEEE Trans. Signal Process.}, vol.~65, no.~16, pp. 4293--4308, 2017.

\bibitem{10571841}
Y.~Zou, Y.~Zhou, X.~Chen, and Y.~C. Eldar, ``Proximal gradient-based unfolding for massive random access in {IoT} networks,'' \emph{IEEE Trans. Wireless Commun.}, vol.~23, no.~10, pp. 14\,530--14\,545, 2024.

\bibitem{10292744}
X.~Dang, W.~Xiang, L.~Yuan, Y.~Yang, E.~Wang, and T.~Huang, ``Deep unfolding scheme for grant-free massive-access vehicular networks,'' \emph{IEEE Trans. Intell. Transp. Syst.}, vol.~24, no.~12, pp. 14\,443--14\,452, 2023.

\bibitem{Heaton_Chen_Wang_Yin_2023}
H.~Heaton, X.~Chen, Z.~Wang, and W.~Yin, ``Safeguarded learned convex optimization,'' \emph{Proceedings of the AAAI Conf. on Artificial Intelligence}, vol.~37, no.~6, pp. 7848--7855, Jun. 2023.

\bibitem{9605579}
Y.~Bai, W.~Chen, B.~Ai, Z.~Zhong, and I.~J. Wassell, ``Prior information aided deep learning method for grant-free {NOMA in mMTC},'' \emph{IEEE J. Sel. Areas Commun.}, vol.~40, no.~1, pp. 112--126, 2022.

\bibitem{10304065}
Y.~Bai, W.~Chen, B.~Ai, and P.~Popovski, ``Deep learning for asynchronous massive access with data frame length diversity,'' \emph{IEEE Trans. Wireless Commun.}, vol.~23, no.~6, pp. 5529--5540, 2024.

\bibitem{10615768}
T.~Li, Y.~Jiang, Y.~Huang, P.~Zhu, F.-C. Zheng, and D.~Wang, ``Model-based deep learning for massive access in mmwave cell-free massive mimo system,'' in \emph{2024 IEEE Int. Conf. Commun. Workshops (ICC Workshops)}, 2024, pp. 828--833.

\bibitem{10192549}
Z.~Ma, W.~Wu, F.~Gao, and X.~Shen, ``Model-driven deep learning for non-coherent massive machine-type communications,'' \emph{IEEE Trans. Wireless Commun.}, vol.~23, no.~3, pp. 2197--2211, 2024.

\bibitem{8807257}
Y.-P. Hsu, E.~Modiano, and L.~Duan, ``Scheduling algorithms for minimizing age of information in wireless broadcast networks with random arrivals,'' \emph{IEEE Trans. Mob. Comput.}, vol.~19, no.~12, pp. 2903--2915, 2020.

\bibitem{9488702}
X.~Chen, X.~Liao, and S.~S. Bidokhti, ``Real-time sampling and estimation on random access channels: Age of information and beyond,'' in \emph{IEEE INFOCOM 2021 - IEEE Conf. on Computer Commun.}, 2021, pp. 1--10.

\bibitem{8933047}
I.~Kadota and E.~Modiano, ``Minimizing the age of information in wireless networks with stochastic arrivals,'' \emph{IEEE Trans. Mob. Comput.}, vol.~20, no.~3, pp. 1173--1185, 2021.

\bibitem{9162973}
H.~Chen, Y.~Gu, and S.-C. Liew, ``Age-of-information dependent random access for massive {IoT} networks,'' in \emph{IEEE INFOCOM 2020 - IEEE Conf. on Computer Commun. Workshops (INFOCOM WKSHPS)}, 2020, pp. 930--935.

\bibitem{9377549}
O.~T. Yavascan and E.~Uysal, ``Analysis of slotted {ALOHA} with an age threshold,'' \emph{IEEE J. Sel. Areas Commun.}, vol.~39, no.~5, pp. 1456--1470, 2021.

\bibitem{10615953}
Y.~Zhu, Y.~Zhu, A.~Gong, Y.~Lin, and Y.~Zhang, ``Age-of-information dependent random access for periodic updating,'' in \emph{2024 IEEE Int. Conf. Commun. Workshops (ICC Workshops)}, 2024, pp. 1189--1194.

\bibitem{10461010}
B.~Yu, Y.~Cai, D.~Wu, C.~Dong, R.~Zhang, and W.~Wu, ``Optimizing age of information for uplink cellular internet of things with random access,'' \emph{IEEE Internet Things J.}, vol.~11, no.~11, pp. 20\,300--20\,313, 2024.

\bibitem{10042422}
J.~Wang, J.~Yu, X.~Chen, L.~Chen, C.~Qiu, and J.~An, ``Age of information for frame slotted {ALOHA},'' \emph{IEEE Trans. Commun.}, vol.~71, no.~4, pp. 2121--2135, 2023.

\bibitem{10123405}
Y.~Huang, J.~Jiao, Y.~Wang, X.~Zhang, S.~Wu, R.~Lu, and Q.~Zhang, ``Age of information minimization for frameless {ALOHA} in grant-free massive access,'' \emph{IEEE Trans. Wireless Commun.}, vol.~22, no.~12, pp. 9778--9792, 2023.

\bibitem{9791264}
X.~Chen, K.~Gatsis, H.~Hassani, and S.~S. Bidokhti, ``Age of information in random access channels,'' \emph{IEEE Trans. Inf. Theory}, vol.~68, no.~10, pp. 6548--6568, 2022.

\bibitem{9365698}
Z.~Chen, N.~Pappas, E.~Björnson, and E.~G. Larsson, ``Optimizing information freshness in a multiple access channel with heterogeneous devices,'' \emph{IEEE Open J. of the Commun. Society}, vol.~2, pp. 456--470, 2021.

\bibitem{10438833}
H.~Wang, Y.~Wang, X.~Xie, and M.~Li, ``A scheduling scheme for minimizing age under delay tolerance in iot systems with heterogeneous traffic,'' \emph{IEEE Internet Things J.}, vol.~11, no.~9, pp. 16\,902--16\,914, 2024.

\bibitem{10495345}
Q.~Zhang, Z.~Xu, X.~Lan, J.~Chen, J.~He, W.~Ma, and Q.~Chen, ``Optimal age of information and throughput scheduling in heterogeneous traffic wireless physical-layer security communications,'' \emph{IEEE Internet Things J.}, vol.~11, no.~13, pp. 23\,644--23\,660, 2024.

\bibitem{10260312}
Z.~Ding, R.~Schober, and H.~V. Poor, ``Impact of {NOMA} on age of information: A grant-free transmission perspective,'' \emph{IEEE Trans. Wireless Commun.}, vol.~23, no.~5, pp. 3975--3989, 2024.

\bibitem{9500896}
C.~Guo, S.~Wu, Z.~Deng, J.~Jiao, N.~Zhang, and Q.~Zhang, ``Age-optimal power allocation policies for {NOMA} and hybrid {NOMA/OMA} systems,'' in \emph{ICC 2021 - IEEE Int. Conf. Commun.}, 2021, pp. 1--6.

\bibitem{10722857}
Y.~Wu, F.~Ke, Y.~Loong~Lee, and D.~Li, ``Age-of-information optimization for {NOMA}-assisted wireless backscatter networks,'' \emph{IEEE Wireless Commun. Lett.}, vol.~13, no.~12, pp. 3623--3627, 2024.

\bibitem{8000687}
Y.~Sun, E.~Uysal-Biyikoglu, R.~D. Yates, C.~E. Koksal, and N.~B. Shroff, ``Update or wait: How to keep your data fresh,'' \emph{IEEE Trans. Inf. Theory}, vol.~63, no.~11, pp. 7492--7508, 2017.

\bibitem{7331960}
Y.~Beyene, C.~Boyd, K.~Ruttik, C.~Bockelmann, O.~Tirkkonen, and R.~Jäntti, ``Compressive sensing for {MTC} in new {LTE} uplink multi-user random access channel,'' in \emph{AFRICON 2015}, 2015, pp. 1--5.

\bibitem{8961111}
M.~Ke, Z.~Gao, Y.~Wu, X.~Gao, and R.~Schober, ``Compressive sensing-based adaptive active user detection and channel estimation: Massive access meets massive {MIMO},'' \emph{IEEE Trans. Signal Process.}, vol.~68, pp. 764--779, 2020.

\bibitem{9163053}
D.~C. Atabay, E.~Uysal, and O.~Kaya, ``Improving age of information in random access channels,'' in \emph{IEEE INFOCOM 2020 - IEEE Conf. on Computer Commun. Workshops (INFOCOM WKSHPS)}, 2020, pp. 912--917.

\bibitem{0901.3403}
D.~Baron, M.~F. Duarte, M.~B. Wakin, S.~Sarvotham, and R.~G. Baraniuk, ``Distributed compressive sensing,'' 2009.

\bibitem{9174792}
Y.~Cui, S.~Li, and W.~Zhang, ``Jointly sparse signal recovery and support recovery via deep learning with applications in {MIMO}-based grant-free random access,'' \emph{IEEE J. Sel. Areas Commun.}, vol.~39, no.~3, pp. 788--803, 2021.

\bibitem{9754266}
P.~Wu and J.~Cheng, ``Deep unfolding basis pursuit: Improving sparse channel reconstruction via data-driven measurement matrices,'' \emph{IEEE Trans. Wireless Commun.}, vol.~21, no.~10, pp. 8090--8105, 2022.

\bibitem{CRMATH_2008__346_9-10_589_0}
E.~J. Cand\`es, ``\BIBforeignlanguage{en}{The restricted isometry property and its implications for compressed sensing},'' \emph{\BIBforeignlanguage{en}{Comptes Rendus. Math\'ematique}}, vol. 346, no. 9-10, pp. 589--592, 2008.

\bibitem{kingma2014adam}
D.~P. Kingma, ``Adam: A method for stochastic optimization,'' \emph{arXiv preprint arXiv:1412.6980}, 2014.

\bibitem{NEURIPS2018}
X.~Chen, J.~Liu, Z.~Wang, and W.~Yin, ``Theoretical linear convergence of unfolded ista and its practical weights and thresholds,'' in \emph{Advances in Neural Information Process. Systems}, vol.~31, 2018.

\end{thebibliography}

\vfill

\end{document}